\documentclass[acmsmall]{acmart}
\AtBeginDocument{%
  \providecommand\BibTeX{{%
    \normalfont B\kern-0.5em{\scshape i\kern-0.25em b}\kern-0.8em\TeX}}}
\acmJournal{TKDD}

\usepackage{algorithm}
\usepackage{algorithmic}
\usepackage{bm}
\usepackage{ bbold }
\usepackage{multirow}
\usepackage{caption}
\usepackage{subfigure}
\usepackage{booktabs}
\usepackage{float}
\usepackage{booktabs}
\usepackage{enumitem}

\usepackage{color}
\usepackage{amsmath}
\usepackage{ dsfont }
\usepackage{ upgreek }
\usepackage{booktabs}
\usepackage{multirow}
\usepackage{xcolor}
\usepackage{colortbl}
\usepackage{subfigure}
\usepackage{threeparttable}

\newtheorem{example}{Example}
\newtheorem{definition}{Definition}
\newtheorem{theorem}{Theorem}
\newtheorem{lemma}{Lemma}



\begin{document}
\title{NSPG-Miner: Mining Repetitive Negative Sequential Patterns}

\author{Yan Li}
\email{lywuc@163.com}
\affiliation{%
	\institution{School of Economics and Management, Hebei University of Technology}
	\city{Tianjin}
	\country{China}
	\postcode{300401}
}
\author{Zhulin Wang}
\email{1299056565@qq.com}
\affiliation{%
	\institution{School of Artificial Intelligence, Hebei University of Technology}
	\city{Tianjin}
	\country{China}
	\postcode{300401}
}

\author{Jing Liu}
\email{liujing@scse.hebut.edu.cn}
\affiliation{%
	\institution{School of Artificial Intelligence, Hebei University of Technology}
	\city{Tianjin}
	\country{China}
	\postcode{300401}
}

\author{Lei Guo}
\email{guoshengrui@163.com}
\affiliation{%
	\institution{State Key Laboratory of Reliability and Intelligence of Electrical Equipment, Hebei University of Technology}
	\city{Tianjin}
	\country{China}
	\postcode{300401}
}

\author{Philippe Fournier-Viger}
\email{philfv@szu.edu.cn}
\affiliation{%
	\institution{College of Computer Science and Software Engineering, Shenzhen University}
	\city{Shenzhen}
	\country{China}
	\postcode{518061}
}

\author{Youxi Wu}
\authornote{Corresponding authors}
\email{wuc567@163.com}
\affiliation{%
	\institution{School of Artificial Intelligence, Hebei University of Technology}
	\city{Tianjin}
	\country{China}
 \institution {and Hebei Key Laboratory of Big Data Computing}
	\country{China}
	\postcode{300401}
}
\author{Xindong Wu}
\email{xwu@hfut.edu.cn}
\affiliation{%
	\institution{the Key Laboratory of Knowledge Engineering with Big Data (the Ministry of Education of China), Hefei University of Technology}
	\city{Hefei}
	\country{China}
	\postcode{230009}
}
\begin{abstract}
Sequential pattern mining (SPM) with gap constraints (or repetitive SPM  or tandem repeat discovery in bioinformatics) can find frequent repetitive subsequences satisfying gap constraints, which are called positive sequential patterns with gap constraints (PSPGs). However, classical SPM with gap constraints cannot find the frequent missing items in the PSPGs. To tackle this issue, this paper explores negative sequential patterns with gap constraints (NSPGs). We propose an efficient NSPG-Miner algorithm that can mine both frequent PSPGs and NSPGs simultaneously. To effectively reduce candidate patterns, we propose a pattern join strategy with negative patterns which can generate both positive and negative candidate patterns at the same time. To calculate the support (frequency of occurrence) of a pattern in each sequence, we explore a NegPair algorithm that employs a key-value pair array structure to deal with the gap constraints and the negative items simultaneously and can avoid redundant rescanning of the original sequence,  thus improving the efficiency of the algorithm. To report the performance of NSPG-Miner, 11 competitive algorithms and 11 datasets are employed. The experimental results not only validate the effectiveness of the strategies adopted by NSPG-Miner, but also verify that NSPG-Miner can discover more valuable information than the state-of-the-art algorithms. Algorithms and datasets can be downloaded from https://github.com/wuc567/Pattern-Mining/tree/master/NSPG-Miner.
\end{abstract}


\maketitle

\section{Introduction}\label{section1}
Sequential pattern mining (SPM) can be used to discover valuable subsequences (called patterns) from a large amount of data \cite {2_Gan2019, dse2024}. SPM has been widely used in many fields, such as customer purchase recommendation \cite{litkde2023}, biological information mining \cite{3_Zhang2020}, social security fraud discovery \cite{4_Okolica2020}, outlying sequence data analysis \cite {outlying}, feature extraction for sequence classification \cite {wucontrast, zhangclassification} and clustering \cite{litkde2024}, click-stream sequences analysis \cite{clickinf}, erasable pattern mining for production optimization \cite {unilerasable}, time-sensitive data stream analysis \cite { unilerasable2}, big data analysis \cite {Siesta}, missing item analysis \cite {dong2023, dong2018}, and medical-order sequence variants analysis\cite {medical}. To meet the diverse needs of users, various SPM methods have been proposed. For example, to discover frequent trends in time series, order-preserving SPM methods were explored \cite {wu2022oppminer, wuopr}.  To mine some lower frequency but more valuable patterns, high-utility mining \cite {10_Gan2021, utilitygan}, high-utility pattern mining with negative unit profits \cite {unil2022}, and high-average-utility mining \cite {unil-utility2} methods were investigated. 

However, the classic SPM methods only focus on whether a pattern occurs in a sequence, but do not pay attention to the number of occurrences of the pattern repetition in the sequence. To deal with this problem, SPM with gap constraints methods were studied \cite{33_Zhang2007, 32_Wu2014, 34_Wang2019}, where a pattern with gap constraints is denoted by $\mathbf p = p_{1}[a,b]p_{2}\ldots[a,b]p_{j}\ldots[a,b]p_{m}$, where $a$ and $b$ stand for the minimum and maximum wildcards between $ p_{j-1} $ and $p_{j}$, respectively~\cite{24_Li2021}. To meet different requirements, several variations of SPM with gap constraints have been proposed. For example, to extract keyphrase in text, one-off SPM with gap constraints was explored \cite{27_Xie2017}. To discover patterns that users are more interested in, items are categorized according to three interest levels, and three-way SPM with gap constraints \cite{12_Min2020,16_Wu2022} were explored. To discover pyramid fraud patterns \cite {pyramid2021}, sequential, time interval, and one-off constraints were considered simultaneously.  For sequence classification \cite {hetkde}, top-k contrast SPM with self-adaptive gap was proposed to mine contrast patterns as classification features \cite {coppminer, contrastzengyouhe}. To discover the frequent patterns with common prefix subpattern, co-occurrence pattern mining methods were explored \cite{guopaa, tmis2024}.

Current methods for SPM with gap constraints mainly mine subsequences that frequently repetitively occur in a sequence \cite{33_Zhang2007, 32_Wu2014, 34_Wang2019}, that is, positive sequential patterns with gap constraints (PSPGs), where the gap constraints do not care about the items in the gap. Therefore, it is difficult for users to know which items are ignored or do not occur through positive patterns, called negative patterns \cite {21_Gao2020,44_Guyet2020}. This means that current SPM with gap constraints methods are unable to discover more potentially important information. An illustrative example is given below. 

\begin{example}\label{exam1}
{ Suppose there is a sequence $ \mathbf s =$  $d_1d_2d_3d_4d_5d_6d_7$ = baacaac that is an activity record of a bank card user. The items ``a'', ``b'', and ``c'' represent the user activities of receiving money through bank transfer, paying bills or purchases, and sending money via bank transfer, respectively. Suppose the threshold $ \rho $ is four. }

{If we only mine PSPGs, we know that $ \mathbf p_{1}$ = a[0,1]c is a PSPG, since there are four occurrences of pattern $ \mathbf p_{1} $ in sequence $ \mathbf s $: $<2,4>, <3,4>, <5,7>, <6,7>$, which is no less than the threshold $ \rho $. Taking $<2,4>$ as an example, we know that $s_2$ and $s_4$ are ``a'' and ``c'', respectively, and there is one item between $d_2$ and $d_4$, which satisfies [0,1] gap constraint. Thus, $<2,4>$ is an occurrence of pattern a[0,1]c. Hence, $ \mathbf p_{1}$ = a[0,1]c is a PSPG.} 

{If we mine negative sequential patterns with gap constraints (NSPGs) at the same time, we can also find that there is no item ``b'' between positions 2 and 4, 3 and 4, 5 and 7, and 6 and 7 of $ \mathbf s $, i.e., pattern $ \mathbf p_{2}$ = a[0,1]$^{\overline {b}}$c also occurs four times in $ \mathbf s $. Therefore, $ \mathbf p_{2}$ = a[0,1]$^{\overline {b}}$c is an NSPG.}

{Therefore, compared with pattern a[0,1]c, pattern a[0,1]$^{\overline {b}}$c is more meaningful. The reasons are as follows. Pattern a[0,1]$^{\overline {b}}$c is frequent which indicates that its positive sequence pattern a[0,1]c is also frequent, since the occurrence of a negative pattern is more strict than that of its corresponding positive pattern. Moreover, we will prove that the support rate of an NSPG is not greater than that of its corresponding positive pattern in Theorem \ref {theo1}. In addition, pattern a[0,1]$^{\overline {b}}$c indicates that the subsequence without item ``b'' between ``a'' and ``c'' is also frequent.}
\end{example}

In Example \ref{exam1}, pattern $ \mathbf p_{1}$ = a[0,1]c shows that the user has transferred money out shortly after receiving funds, which is in accordance with the regular banking business. 
Pattern $ \mathbf p_{1} $ is a PSPG, which indicates that a[0,1]c is a frequent pattern, but it cannot indicate which items are missing, while pattern $ \mathbf p_{2}$ = a[0,1]$^{\overline {b}}$c is an NSPG that requires ``a'' and ``c'' to satisfy the gap constraints, and there should be no item ``b'' between ``a'' and ``c''. We find that the user has no payment behavior during the period of time after the funds are transferred in, which indicates that there is a possibility of money laundering. Therefore, NSPGs have high applicability and research value. They can find information that PSPGs cannot capture and thus provide more comprehensive information to the user. 

To discover positive and more meaningful negative patterns, inspired by negative SPM \cite{43_Dong2020}, based on the classical SPM with gap constraints \cite{33_Zhang2007, 32_Wu2014, 34_Wang2019}, we propose the concept of negative sequential pattern with gap constraints (NSPG), which refers to sequential patterns with gap constraints that contain both positive and negative items. The main contributions are as follows:

\begin{enumerate}
\item To discover frequent PSPGs and NSPGs, this paper studies a new issue named NSPG mining and proposes an NSPG-Miner algorithm, which performs two key tasks: candidate pattern generation and support calculation.

\item For candidate pattern generation, we propose a pattern join strategy with negative patterns to generate candidate patterns, which can generate both positive and negative candidate patterns and effectively reduce the number of candidate patterns.

\item To improve support calculation efficiency, we employ a key-value pair array structure that can avoid redundant rescanning of the original sequence dataset to improve the calculation efficiency of the support of candidate patterns.

\item  Experimental results demonstrate the effectiveness of NSPG-Miner and that it can find more meaningful negative patterns than the state-of-the-art algorithms.
\end{enumerate}

The structure of this paper is as follows. Section \ref{section2} discusses the related work. Section \ref{section3} gives the relevant definitions. Section \ref{section4} proposes NSPG-Miner. Section \ref{section5} reports the performance of NSPG-Miner on many real datasets. Section \ref{section6} summarizes the conclusions.

\section{Related Work}\label{section2}
In this section, we will summarize the related work concerning two aspects: SPM with gap constraints and negative SPM. 

\subsection{SPM with gap constraints}
Sequential pattern mining (SPM)  is an important knowledge discovery method \cite{pmdb2022,wu2022tmis}. A variety of SPM methods have been proposed to meet different requirements, such as rare pattern mining \cite {rarepattern,rarepattern2}, utility pattern mining \cite{hanpminer, unil-utility1, 6_Gan2021}, episode pattern mining \cite{13episodepattern, 14episoderule}, closed or maximal SPM  \cite{29_Wu2020, 30_Li2021}, process pattern mining \cite {processmining1, processmining2}, distributed pattern mining \cite {distributedmining}, spatial co-location pattern mining \cite{Wang2022colocation, Wang2023}, negative SPM \cite {dong2023, dong2018}, and SPM with gap constraints \cite{tkdd2012}. A disadvantage of traditional SPM methods is that these methods only pay attention to whether a pattern occurs in a sequence but disregard that a pattern may occur multiple times in it. Therefore, some important patterns may be missed using traditional SPM.  To overcome this issue, SPM with gap constraints was proposed \cite{33_Zhang2007}. 

SPM with gap constraints has many different names, such as repetitive SPM \cite {yongxintong} or tandem repeat discovery in bioinformatics \cite {tandemrepeat}. More importantly, SPM with gap constraints is more difficult than classical SPM, since the calculation of the pattern support in sequences becomes a pattern matching issue \cite{scis2017}. Moreover, SPM with gap constraints has four different methods to calculate the number of occurrences: no-condition  \cite{12_Min2020}, one-off condition (or one-off SPM) \cite{45_Wu2021,oneoff2}, nonoverlapping condition (or nonoverlapping SPM) \cite{14_Wu2018}, and disjoint condition \cite{disjoint2021}. The case of no-condition means that each item can be reused. The one-off condition means that each item can be used at most once \cite{28_Wu2021}. The nonoverlapping condition means that each item cannot be reused by the same $p_j$, but can be reused by different $p_j$ \cite{gengtkde2024}. The disjoint condition means that the maximum position of an occurrence should be less than the minimum position of the next occurrence. 

Among the above four methods, the no-condition was proposed first \cite {33_Zhang2007}. More importantly, the no-condition is arguably the most popular method \cite{nocondition2021}, since the support calculation of the no-condition can be regarded as calculating the number of all occurrences and it is easy to be understood, while the support calculations of other three methods can be regarded as only counting some specific occurrences.

Although the above SPM with gap constraints methods can find interesting patterns in various situations, they ignore the missing items in sequences. Therefore, it is easy to overlook potentially important information.

\subsection{Negative SPM}
{To discover missing items, Ouyang and Huang \cite {ouyang2007} first proposed three forms of negative SPM in transaction databases: $<$$\overline {X}$, Y$>$, $<$X, $\overline {Y}$$>$, and $<$$\overline {X}$, $\overline {Y}$$>$. One of the drawbacks of this study is that it limits the position of negative items in negative sequence pattern mining. To address this problem, Hsueh et al. \cite{39_Hsueh2008} designed a PNSP algorithm to mine negative sequential patterns in which the location of negative items is more flexible. Based on PNSP, Zheng et al. \cite{40_Zheng2009} introduced a GSP-like method, called neg-GSP, to mine negative sequential patterns. Both PNSP and neg-GSP do not satisfy the Apriori property, and they need to mine positive sequential patterns before mining negative sequential patterns. Therefore, the running performances of PNSP and neg-GSP need to be improved.} 

Note that in NSP mining, there is currently no unified definition of negative containment due to the concealment of non-occurrence items \cite{dong2019}. Therefore, various negative SPM were proposed \cite {hannegative, 18_Wang2021}. Among them,  e-NSP \cite{41_Cao2016},  e-RNSP \cite{43_Dong2020}, and ONP-Miner \cite{onpminer} are three state-of-the-art methods. e-NSP \cite{41_Cao2016} and e-RNSP  \cite{43_Dong2020} are two very similar methods, since both e-NSP and e-RNSP discover frequent positive patterns at first. Then, they generate negative candidate patterns based on frequent positive patterns and mine frequent negative patterns. The difference between e-NSP and e-RNSP is that the former does not consider the repetition of patterns in the sequence, while the latter does.  For example, the support of pattern ``abc'' in sequence  $ \mathbf s $ =acabcabcac is one according to e-NSP, since e-NSP does not consider the repetition of patterns in the sequence, while the support of pattern ``abc''  in sequence  $ \mathbf s $ =acabcabcac is two according to e-RNSP, since e-RNSP considers the repetition. 

However, ONP-Miner is far different from e-NSP and e-RNSP, since ONP-Miner \cite {onpminer} has gap constraints, while e-NSP \cite{41_Cao2016} and e-RNSP  \cite{43_Dong2020} do not have gap constraints. Due to the consideration of gap constraints, ONP-Miner can discover more valuable negative patterns, and the experimental results showed that ONP-Miner can mine more negative patterns than e-NSP and e-RNSP.

ONP-Miner is a negative SPM with gap constraints under the one-off condition and adopts one-off pattern matching to calculate the support of a pattern. Unfortunately, it was proven that one-off pattern matching with gap constraints is an NP-Hard problem \cite {28_Wu2021, oneoff2}, which means that the support cannot be accurately calculated and can only be approximated calculation using a heuristic strategy. Thus, ONP-Miner is an approximate mining algorithm that may lose some feasible patterns. 


{To tackle the above issues, compared with the previous research, the novelties of this paper are three aspects. }

\begin {enumerate}
\item { To overcome the shortages of ONP-Miner \cite {onpminer}, this paper investigates repetitive negative SPM with gap constraints under no condition which can mine all feasible positive and negative patterns.}

\item {Classical negative SPM methods, such as e-NSP \cite{41_Cao2016} and e-RNSP \cite{43_Dong2020}, discover frequent positive patterns and then generate negative candidate patterns, which generate redundant candidate negative patterns. To effectively reduce the number of candidate patterns, we propose a pattern join strategy with negative patterns to generate candidate patterns based on Theorems 1 and 2, which can simultaneously generate both positive and negative candidate patterns.}

\item { Classical SPM with gap constraints method \cite{32_Wu2014} adopts Incomplete-Nettree structure to calculate the supports of candidate patterns with the same prefix pattern, which requires scanning the whole sequence for each prefix pattern. To improve the efficiency, we adopt the key-value arrays of prefix and suffix subpatterns to calculate the supports of both positive and negative candidate patterns with gap constraints that can avoid redundant rescanning of the original sequence dataset, thus improving the calculation efficiency of the support of candidate patterns.}

\end {enumerate}

\section{Problem definition}\label{section3}
\begin{definition}
{(Sequence and sequence database)}    	A sequence with a length of $l$ can be represented by $\mathbf s = d_{1}d_{2} \ldots d_{l}$. A sequence database consists of $n$ sequences and is expressed as $SDB = {\mathbf s_{1},\mathbf s_{2} \ldots \mathbf s_{n}}$. The sum of the lengths of all sequences in $SDB$ is the length of $SDB$, which is expressed as $L$. A collection of different items in SDB is expressed as $\Sigma$. $|\Sigma|$ represents the number of items in $\Sigma$.
\end{definition}
    
{Example \ref {exam2} illustrates the concepts of sequence and sequence database.}    

\begin{example}\label{exam2}
Suppose we have a sequence database $SDB=\{\mathbf s_{1}, \mathbf s_{2}\}$, $\mathbf s_{1} =$ baacaac, and $\mathbf s_{2} = $ ababccbb. The lengths of $\mathbf s_{1}$ and $\mathbf s_{2}$ are $l_{1} = 7$ and $l_{2} = 8$, respectively. Hence, the length of SDB is $L = l_{1} + l_{2} =15$. The sequence database items  $\Sigma=\{$a,b,c$\}$, and $|\Sigma| = 3$.
    \end{example}
    
\begin{definition}\label{def2} 
{(Negative item)} 	  Given an item $e$ $(e\in\Sigma)$, its corresponding negative item is $\overline {e}$. The negative item $\overline {e}$ means that there is no item $e$ between $d_{j}$ and $d_{k}$, where $j<k$.
    \end{definition}

{Example \ref {exam3} illustrates the concept of negative item.}    
    
\begin{example}\label{exam3}
  Suppose we have a sequence $\mathbf s_{1} = d_{1}d_{2}\ldots d_{7} = baacaac$, $\Sigma = \{$a,b,c$\}$. There are no item ``b'' and ``c'' between $d_{5}$ and $d_{7}$ in $\mathbf s_{1}$. Thus, negative items $\overline {b}$ and $\overline {c}$ exist between $d_{5}$ and $d_{7}$ in $\mathbf s_{1}$.
\end{example}
    
\begin{definition}\label{def3}  
{(NSPG and PSPG)}
 An NSPG can be expressed as $\mathbf p = p_{1}[a,b]^{\overline {e_{1}}}p_{2}\ldots p_{j}[a,b]^{\overline {e_{j}}}p_{j+1}$ $\ldots [a,$ $b]^{\overline {e_{m-1}}}p_{m}$, where $p_{j}$ is a positive item, $p_{j}\in\Sigma$, $a$ and $b$ stand for the minimum and maximum wildcards between $ p_{j-1} $ and $p_{j}$, respectively, and the number of positive items in the pattern is called the length of $\mathbf p$. $\overline {e_{j}} $ is a negative item: $e_{j}\in\Sigma$ or $e_{j}=null$. If all $e_{j}$ are null, then pattern $\mathbf p$ is a PSPG.
\end{definition}

{Example \ref {exam4} illustrates the concepts of NSPG and PSPG.}    
        
\begin{example}\label{exam4}
Suppose we have a pattern $\mathbf p = p_{1}[a,b]^{\overline {e_{1}}}p_{2}[a,b]^{\overline {e_{2}}}p_{3}=$ a[0,1]a[0,1]$^{\overline {b}}$c, which is an NSPG.  $\mathbf p$ contains three positive items: a, a, and c. Hence, the length of $\mathbf p$ is 3. There is no negative item in a[0,1]a, since $\overline {e_{1}}$ = null. $\overline {e_{2}} = \overline {b}$, which means that there is no item b between $p_{2}$ = a and $p_{3}$ = c. The corresponding PSPG of pattern $\mathbf p$ is a[0,1]a[0,1]c.
\end{example}

\begin{definition}\label{def4}
{ (Offset sequence)  If any two indexes $i_{j-1}$ and $i_{j}$ $(2\leq j\leq l)$ in  $I = <i_{1}, i_{2} \ldots i_{m}>$ satisfy $M \leq i_{j}-i_{j-1}-1 \leq N$, then $I$ is an offset sequence of pattern $\mathbf p$ in sequence $\mathbf s$. The number of all offset sequences of $\mathbf p$ in $\mathbf s$ is represented by $ofs(\mathbf p, \mathbf s)$. The total offset sequences in $SDB$ is represented by $ofs(\mathbf p, SDB)$, and $ofs(\mathbf p, SDB)=\sum\nolimits_{i=1}^nofs(\mathbf p,\mathbf s_{i})$.}
\end{definition}

{Lemma \ref{lemma1} shows the calculation method of offset sequence, and Example \ref {exam5} illustrates the calculation method of offset sequence.}

 \begin{lemma}\label{lemma1}
The total number of offset sequences of $\mathbf p$ in $SDB$ can be calculated according to $ofs(\mathbf p, SDB) = L\times W^{m-1}$, where $L$ is the length of $SDB$, $m$ is the length of $\mathbf p$, and $W = M-N+1$ ~\cite{12_Min2020}.
\end{lemma}
 
\begin{example}\label{exam5}
Suppose we have a pattern $\mathbf p$ = $a[0,1]a[0,1]^{\overline {b}}c$, and the sequence database $SDB$ is the same as in Example \ref {exam2}. Since the length of $\mathbf p$ is $m = 3$, the offset sequence length of $\mathbf p$ is also 3. The total number of offset sequences of $\mathbf p$ in $SDB$ is $ofs(\mathbf p,SDB)=L\times W^{m-1}=15\times(1-0+1)^{3-1}=60$.
\end{example}
    
{Based on the concept of offset sequence, we show the definition of occurrence in Definition \ref {def5}.}

\begin{definition}    \label{def5} 
(Occurrence) 	    Suppose $I = <i_{1}, i_{2} \ldots i_{m}>$ is an offset sequence. If $d_{i_{j}} = p_{j}$ $(1 \leq j \leq m)$ and there is no item $e_{j}$  ($e_j \in \Sigma$) between $d_{i_{j}}$ and $d_{i_{j+1}}$ , then $I$ is an occurrence of $\mathbf p$ in $\mathbf s$.
\end{definition}

{Example \ref {exam6} illustrates the occurrences of a pattern in a sequence.}

\begin{example}\label{exam6}
For a pattern $\mathbf p = p_{1}[a,b]^{\overline {e_{1}}}p_{2}[a,b]^{\overline {e_{2}}}p_{3}$ = $a[0,1]a[0,1]^{\overline {b}}c$ and a sequence $\mathbf s_{1} = d_{1}d_{2}\ldots d_{7}$ = baacaac, $I = <3,5,7>$ is an offset sequence of $\mathbf p$ in $\mathbf s_{1}$, since positions 3 and 5 satisfy the gap constraint $0 \leq 5-3-1 \leq 1$, and so do positions 5 and 7. $I$ is an occurrence of $\mathbf p$ in $\mathbf s_{1}$, since $d_{3} = p_{1}$ = a, $d_{5} = p_{2}$ = a, $d_{7} = p_{3}$ = c, and there is no item $e_{2}$ = b between $d_{5}$ and $d_{7}$. All occurrences of $\mathbf p$ in $\mathbf s_{1}$ are shown in Figure \ref{fig2}.
\end{example}

\begin{figure}
\centering
\includegraphics [width=0.5\textwidth] {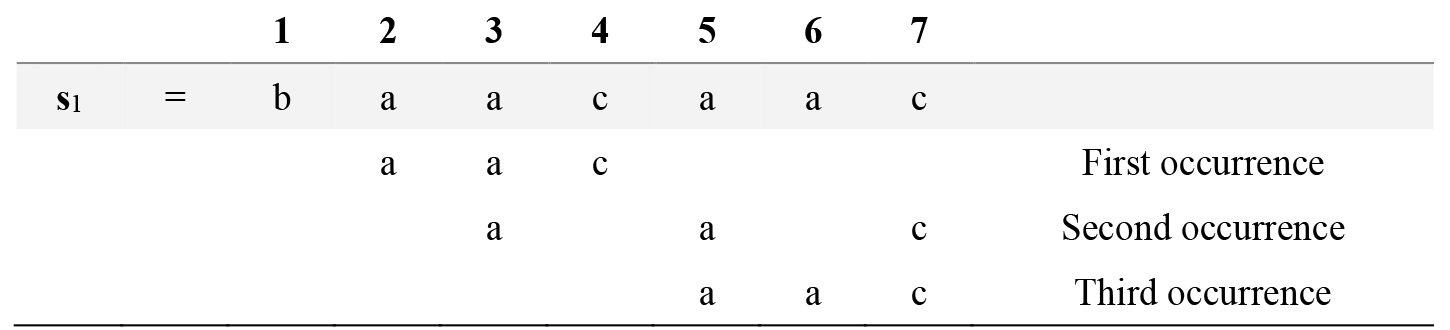}
\caption{All occurrences of pattern $ \mathbf p $ in sequence $ \mathbf s_{1} $}
\label{fig2}
\end{figure}

{Based on the concepts of offset sequence and occurrence, we show the definitions of support, support rate, and frequent pattern in Definition \ref {def6}.}

\begin{definition}\label{def6} 
{(Support, support rate, and frequent pattern)}    	The total number of occurrences of pattern $ \mathbf p $ in $\mathbf s$ is called the support, represented by $sup(\mathbf p, \mathbf s)$ in $\mathbf s$, and $sup(\mathbf p, SDB)$ is the sum of the supports of each pattern in $SDB$, i.e., $sup(\mathbf p, SDB)= \sum\nolimits_{i=1}^nsup(\mathbf p,\mathbf s_{i})$. The support rate of $ \mathbf p $ in $SDB$, represented by $r(\mathbf p, SDB)$, is the rate of its support and offset, i.e., $r(\mathbf p, SDB) = sup(\mathbf p, SDB)/ ofs(\mathbf p, SDB)$. If $r(\mathbf p, SDB)$ is no less than the given threshold value $ \rho $, then $\mathbf p$ is a frequent pattern.
\end{definition}

{Example \ref {exam7} illustrates the concepts of support, support rate, and frequent pattern.}
    
\begin{example}\label{exam7}
Suppose we have a pattern $\mathbf p$ = $a[0,1]a[0,1]^{\overline {b}}c$ and a threshold $ \rho=0.13 $, and the sequence database $SDB$ is the same as in Example \ref{exam2}. As shown in Figure \ref {fig2}, $\mathbf p$ occurs three times in $\mathbf s_{1}$, i.e., $sup(\mathbf p, \mathbf s_{1}) = 3$. Similarly, $\mathbf p$ does not occur in $\mathbf s_{2}$, i.e., $sup(\mathbf p, \mathbf s_{2}) = 0$. Thus, $sup(\mathbf p, SDB) = 3+0 = 3$. From Example \ref{exam5}, we know that $ofs(\mathbf p, SDB) = 60$. Therefore, $r(\mathbf p, SDB) = 3/60 = 0.05$. Hence, $\mathbf p$ is not a frequent pattern, since $r(\mathbf p,SDB) < \rho$.
\end{example}
 
\begin{definition}\label{def7} 
{(Our problem)}	    Given a sequence database, gap constraints, and a threshold, our goal is to mine all frequent PSPGs and NSPGs.
\end{definition}

{Example \ref {exam8} illustrates our problem and its results.}

\begin{example}\label{exam8}
Suppose we have a sequence database $SDB$ (as in Example \ref {exam2}), a gap constraint [0,1], and a threshold $\rho = 0.13$. The frequent PSPGs and NSPGs are $\{$$a$, $b$, $c$, $a[0,1]a$, $a[0,1]^{\overline {a}}a$, $a[0,1]c$, $a[0,1]^{\overline {b}}c$, $a[0,1]^{\overline {c}}c$$\}$.
\end{example}

\section{Algorithm Design} \label{section4}
In this section, we propose NSPG-Miner, which performs two key operations: candidate pattern generation and support calculation.
 The framework of  NSPG-Miner is shown in Figure \ref{frame}. Section \ref{sect4.1} describes the candidate pattern generation strategy. Section \ref{sect4.2} presents the NegPair algorithm to calculate the support. Section \ref{sect4.3} proposes the NSPG-Miner algorithm to mine all frequent PSPGs and NSPGs.

\begin{figure}[h]
		\centering
		\includegraphics [width=0.7\textwidth]  {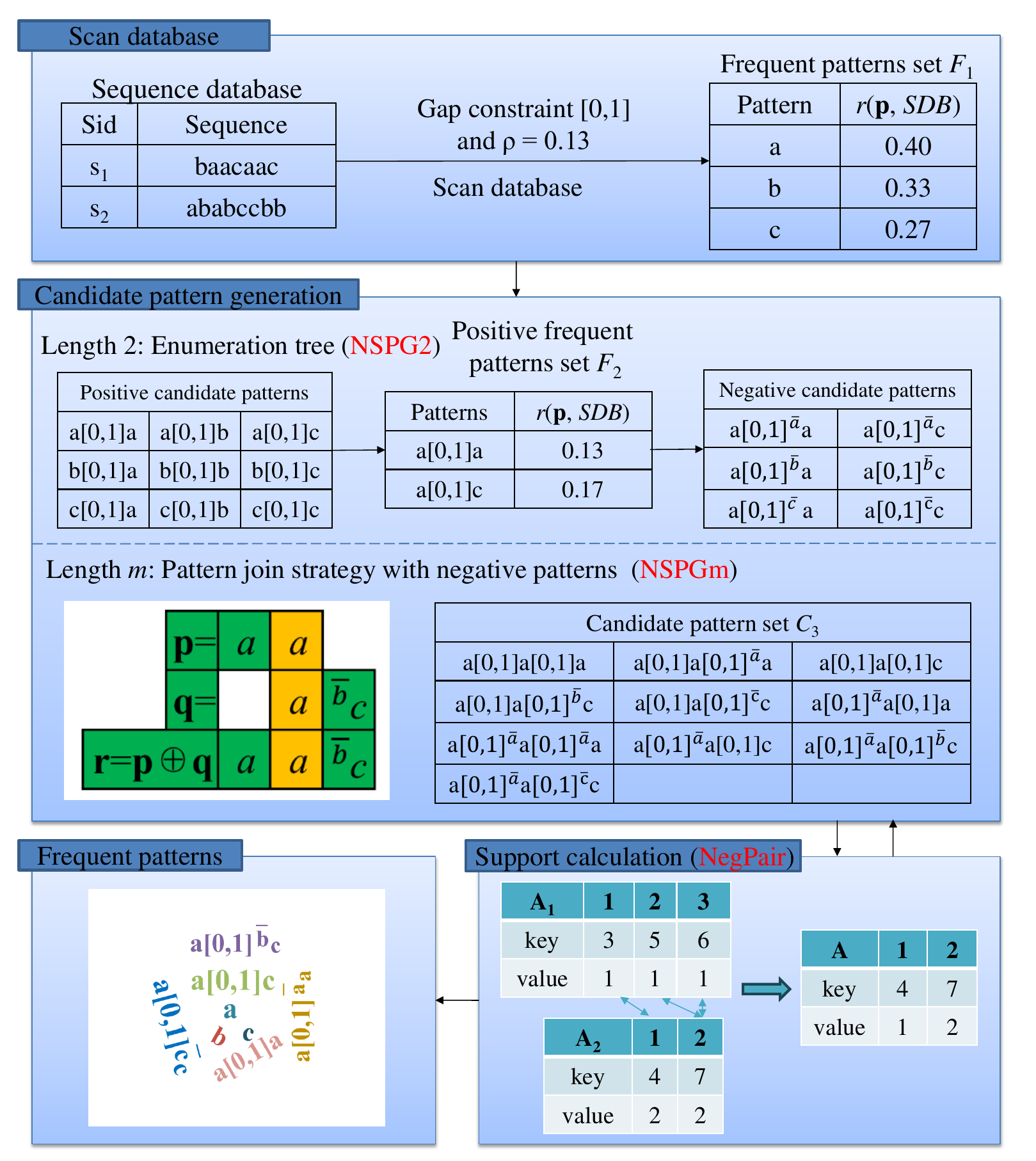}
		\caption{{Framework of NSPG-Miner. NSPG-Miner has two essential parts: candidate pattern generation and support calculation. To effectively generate candidate patterns with length $m$ ($m>$ 2), we propose a pattern join strategy with negative patterns which can generate positive and negative candidate patterns simultaneously. To improve the efficiency of support calculation, we employ a key-value array structure to avoid redundant scanning of the whole sequence.}}
		\label{frame}
	\end{figure}

\subsection{Candidate pattern generation} \label {sect4.1}

In SPM with gap constraints, the number of candidate patterns is one of the main factors affecting the mining efficiency, since we need to calculate the support for each candidate pattern and then determine whether it is frequent or not. Thus, the more candidate patterns, the slower the mining speed. 

{ The enumeration-tree strategy is a commonly used strategy to generate candidate patterns \cite{apin2020,tmis2021}, since this method can generate all possible candidate patterns, thus ensuring the completeness of the mining algorithm. However, this method will generate some redundant positive and negative candidate patterns. The classical negative SPM algorithms, such as e-NSP \cite{41_Cao2016} and e-RNSP \cite{43_Dong2020}, usually adopt the classical pattern join strategy \cite{33_Zhang2007} to generate positive candidate patterns and then discover positive frequent patterns. Finally, these algorithms generate negative candidate patterns by the enumeration method based on these positive frequent patterns. However, this kind of method will also generate some redundant negative candidate patterns. To overcome the shortages of the above methods, we propose a pattern join strategy with negative patterns which can generate positive and negative candidate patterns simultaneously. Some related definitions and theorems are as follows. }

\begin{definition}\label{def8} (Prefix pattern and suffix pattern)
Suppose we have a superpattern $\mathbf p = p_{1}[a,b]^{\overline {e_{1}}}p_{2}$ $\ldots  ^{\overline {e_{m-2}}}p_{m-1}[a,$ $b]$ $^{\overline {e_{m-1}}}p_{m}$. Its prefix pattern is $p_{1}[a,b]^{\overline {e_{1}}}p_{2}\ldots ^{\overline {e_{m-2}}}p_{m-1}$, represented by Prefix$(\mathbf p)$, and its suffix pattern is $p_{2}\ldots ^{\overline {e_{m-2}}}p_{m-1}$ $[a,b]^{\overline {e_{m-1}}}p_{m}$, represented by Suffix$(\mathbf p)$.
\end{definition}

{In order to effectively prune candidate patterns, we propose the following theorems. }

\begin{theorem}
   {The support rate of an NSPG is not greater than that of its corresponding positive pattern.}\label{theo1}  
\end{theorem}

\begin {proof}
{ Suppose we have an NSPG $\mathbf t = p_{1}[a,b]$ $^{\overline {e_{1}}}$ $p_{2} $ $\ldots[a,b]^{\overline {e_{m-1}}}p_{m}$ with a length of $m$,  $e_j \in \Sigma$, and its corresponding positive pattern is $\mathbf p = p_{1}[a,b]p_{2}$ $ \ldots[a,b] p_{m}$. If $I$ = $<i_{1}, i_{2}, \ldots,  i_{m}>$ is an occurrence of pattern $\mathbf t$, then we can safely say that $I$ is also an occurrence of pattern $\mathbf p$, since  there is no item $e_{j}$  ($e_j \in \Sigma$) between $d_{i_{j}}$ and $d_{i_{j+1}}$ according to Definition \ref {def5}, which is more strict than positive pattern. Therefore, $sup(\mathbf t,SDB$) $\leq$ $sup(\mathbf p,SDB$). We know the length of $\mathbf t$ is the same as that of $\mathbf p$. Furthermore, the offset sequence of $\mathbf t$ is the same as that of $\mathbf p$. Hence, $r(\mathbf t,SDB$) $\leq$ $r(\mathbf p,SDB$), i.e., the support rate of an NSPG is not greater than that of its corresponding positive pattern.}
\end {proof}

{According to Theorem \ref {theo1}, if an NSPG is frequent, then its corresponding positive pattern is also frequent. Moreover, if a positive pattern is not frequent, then all its negative patterns are not frequent. Now, we further prove that NSPG mining satisfies the  Apriori property. }

\begin{theorem}
NSPG mining satisfies the  Apriori property.
\label{theorem1}  
\end {theorem}

\begin {proof}
{Suppose we have a pattern $\mathbf p$ with length $m$-1. $\mathbf p$ has a positive superpattern $\mathbf q = \mathbf p[a,b]p_{m}$. Suppose $I = <i_{1}, i_{2}, \ldots, i_{m-1}>$ and $I^{'} = <i_{1}, i_{2}, \ldots, i_{m-1}, i_{m}>$ are the occurrences of $\mathbf p$ and $\mathbf q$ in $\mathbf s$, respectively. We know that $i_{m-1}$ and $i_{m}$ satisfy the gap constraints, i.e., $i_{m}$ is from $i_{m-1}+M+1$ to $i_{m-1}+N+1$. Thus, $I^{'}$ is $W$ times more than $I$ at most, i.e., $sup(\mathbf q, \mathbf s) \leq sup(\mathbf p, \mathbf s)\times W$.  Furthermore, $ \sum\nolimits_{i=1}^nsup(\mathbf q, \mathbf s_{i})$ $\leq $ $\sum\nolimits_{i=1}^nsup(\mathbf p, \mathbf s_{i})\times W $, i.e., $sup(\mathbf q,SDB) \leq  sup(\mathbf p,SDB)\times W$. According to Lemma \ref{lemma1}, we know that $ofs(\mathbf q,SDB) = L\times W^{m-1}$ and  $ofs(\mathbf p,SDB) = L\times W^{m-2}$. Thus, $ ofs(\mathbf q,SDB) = ofs(\mathbf p, SDB)\times W$. Hence,  $ r(\mathbf q,SDB) =\frac{sup(\mathbf q,SDB)}{ofs(\mathbf q,SDB)}\leq \frac{sup(\mathbf p,SDB)}{ofs(\mathbf p,SDB)}=r(\mathbf p,SDB)$. Therefore, if $\mathbf p$ is not a frequent pattern, i.e., $r(\mathbf p,SDB) < \rho$, then its superpattern  $\mathbf q$  is not frequent, since $r(\mathbf q,SDB) \leq r(\mathbf p,SDB) < \rho$.  Moreover, according to Theorem \ref {theo1}, we know that negative pattern  $\mathbf t = \mathbf p[a,b]^{\overline {e_{m-1}}}p_{m}$ is not frequent, since its positive pattern $\mathbf q$ is not frequent. Similarly, we know that if a suffix pattern is not a frequent pattern, then its superpatterns are not frequent patterns, either. Hence, NSPG mining satisfies the Apriori property.}
\end{proof}

Since NSPG mining satisfies the Apriori property, to effectively prune the candidate patterns, this paper proposes a candidate pattern generation method which is called pattern join strategy with negative patterns which is shown in Definition \ref {def9}.

\begin{definition}\label{def9} (Pattern join strategy with negative patterns)	{	Suppose we have two patterns $\mathbf p = p_{1}[a,b]$ $^{\overline {e_{1}}}p_{2}\ldots [a,b]^{\overline {e_{m-2}}}$ $p_{m-1} [a,b]^{\overline {e_{m-1}}}p_{m}$ and $\mathbf q = q_{1}[a,b]^{\overline {f_{1}}}q_{2}\ldots [a,b]^{\overline {f_{m-2}}}q_{m-1}$ $[a,b]^{\overline {f_{m-1}}}q_{m}$. If $ \mathbf r$ = Suffix($\mathbf p$) = Prefix($\mathbf q$), then $\mathbf p$ and $\mathbf q$ can generate a superpattern $\mathbf t=\mathbf p\oplus\mathbf q=p_{1}[a,b]^{\overline {e_{1}}}\mathbf r^{\overline {f_{m-1}}}q_{m}$, where $\mathbf t$ is a candidate pattern with a length of $m+1$. This process is called a pattern join strategy with negative patterns.}
\end{definition}

{Example \ref{exam10} illustrates the principle of pattern join strategy with negative patterns.} 

\begin{example}\label{exam10}
		Suppose we have two patterns, $\mathbf p$ = a[0,1]a and $\mathbf q$ = $a[0,1]^{\overline {b}}c$. Since $\mathbf r$ = Suffix($\mathbf p$) = Prefix$(\mathbf q)$ = a, $\mathbf p$ and $\mathbf q$ can generate a candidate $\mathbf t = \mathbf p\oplus\mathbf q$= $a[0,1]a[0,1]^{\overline {b}}c$, as shown in Figure \ref{fig3}.
	\end{example}
	\begin{figure}[h]
		\centering
		\includegraphics[width=0.12\textwidth]{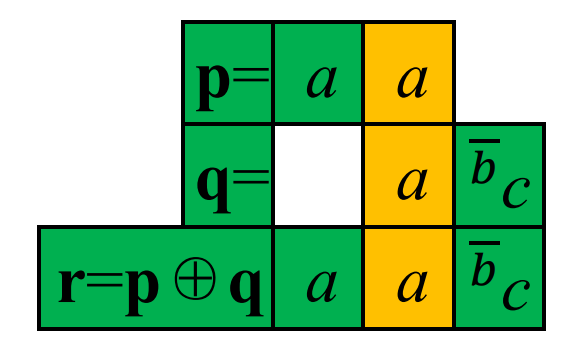}
		\caption{An illustrative example of pattern join strategy with negative patterns}
		\label{fig3}
	\end{figure}

{Example \ref{exam11} illustrates the pattern join strategy with negative patterns has a better performance than other methods.}

\begin{example}\label{exam11}
In Example \ref{exam8}, there are five frequent patterns with a length of two: a[0,1]a, $a[0,1]^{\overline {a}}a$, a[0,1]c, $a[0,1]^{\overline {b}}c$, and $a[0,1]^{\overline {c}}c$. Using the pattern join strategy with negative patterns, ten candidate patterns are generated: a[0,1]a[0,1]a, a[0,1]a[0,1]$^{\overline {a}}$a, a[0,1]a[0,1]c, a[0,1]a[0,1]$^{\overline {b}}$c, a[0,1]a[0,1]$^{\overline {c}}$c, a[0,1]$^{\overline {a}}$a[0,1]a, a[0,1]$^{\overline {a}}$a[0,1]$^{\overline {a}}$a, a[0,1]$^{\overline {a}}$a[0,1]c, a[0,1]$^{\overline {a}}$a[0,1]$^{\overline {b}}$c, and a[0,1] $^{\overline {a}}$a[0,1]$^{\overline {c}}$c.

{ For the enumeration tree strategy, one frequent pattern $\mathbf p$ can generate $|\Sigma|$ positive candidate patterns $\mathbf t = \mathbf p[a,b]q_{m}$ $(q_{m}\in\Sigma)$. Moreover, each positive candidate pattern $\mathbf t$ can further generate $|\Sigma|$ negative candidate patterns  $\mathbf x = \mathbf p[a,b]^{\overline {f_{m-1}}} q_{m}$ $(f_{m-1}\in\Sigma)$. Thus, one frequent pattern $\mathbf p$ can generate $|\Sigma|$ positive candidate patterns and $|\Sigma|\times |\Sigma|$ negative candidate patterns.  In Example \ref{exam8}, there are three items in $\Sigma$, and five frequent patterns with a length of two. Therefore, $5\times(3+3\times3)$=60 positive and negative candidate patterns with a length of three are generated by the enumeration tree strategy.}

{The classical negative SPM algorithms usually generate positive candidate patterns at first. Two positive candidate patterns with a length of three in Example \ref{exam8}  are generated: a[0,1]a[0,1]a and a[0,1]a[0,1]c. Then, negative candidate patterns are generated based on positive patterns. For each gap constraint, $|\Sigma|$ negative candidate patterns can be generated. Therefore, there are $3\times3$=9 negative candidate patterns for each positive candidate pattern with a length of three. There are two gap constraints in a pattern with a length of three.   Hence, this kind of method generates 2+$2\times 3\times 3$=20 positive and negative candidate patterns with a length of three in Example \ref{exam8}.}

 {In summary, the pattern join strategy with negative patterns is superior to the enumeration tree strategy and the classical negative SPM algorithms, since the pattern join strategy with negative patterns generates ten candidate patterns, while the other two methods generate 60 and 20 candidate patterns, respectively. }
\end{example}

\subsection{Support calculation} \label {sect4.2}
In SPM with gap constraints, support calculation is the main time-consuming part. Although the Incomplete-Nettree structure proposed in  \cite{32_Wu2014} can be used to calculate the supports of candidate patterns with the same prefix pattern, this method requires scanning the whole sequence for each prefix pattern. To overcome this shortage, we propose a NegPair algorithm that employs a key-value array structure. The NegPair algorithm uses the key-value arrays of prefix and suffix subpatterns to calculate the supports of both positive and negative candidate patterns with gap constraints. Thus, NegPair does not need to scan the whole original sequence. Hence, the efficiency of the algorithm is improved. The structure of the key-value array is as follows.

Suppose $I = <i_{1}, i_{2} \ldots i_{m}>$ is an occurrence of $\mathbf p$ in $\mathbf s$. $i_{m}$ is the ending position of the occurrence. We use the key-value array structure to calculate the support. Each element in the key-value array has two fields: key and value.  The key is used to store the ending position of $\mathbf p$ in $\mathbf s$, and the value is used to store the number of occurrences with the same ending position. Example \ref {exam12} illustrates the key-value array.
 
\begin{example}\label{exam12}
Suppose we have a pattern $\mathbf t$ = a[0,1]a[0,1]$^{\overline {b}}$c and a sequence $\mathbf s_{1}$ = baacaac. There are three occurrences of $\mathbf t$ in $\mathbf s_{1}$ shown in Figure \ref{fig2}. We know that the key-value pair array $A$ of $\mathbf t$ in $\mathbf s_{1}$ has two elements, $A[1]$ and $A[2]$, which is shown in the right part of Figure \ref{fig4}. $A[1].key = 4$, since the ending position of $I_{1} = <2,3,4>$ is 4. Hence, $A[1].value = 1$. Similarly, $A[2].key = 7$, since the end positions of $I_{2} = <3,5,7>$ and $I_{3} = <5,6,7>$ are both 7, and $A[2].value = 2$. Thus, the sum of the values of all elements in $A$ is the support of $\mathbf t$ in $\mathbf s_{1}$. For example, $A[1].value+A[2].value = 1+2 = 3$ is the support of pattern $\mathbf t$ in $\mathbf s_{1}$. 

\begin{figure}[h]
	\centering
	\includegraphics[width=0.45\textwidth]{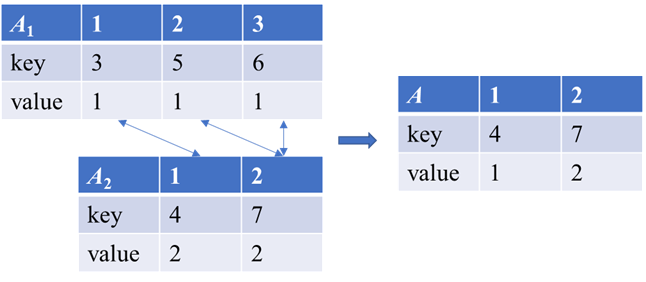}
\caption{Key-value pair arrays $A_1$, $A_2$, and $A$ of patterns  $\mathbf p$ =a[0,1]a, $\mathbf q$=a[0,1](¬b)c, and $\mathbf t$=a[0,1]a[0,1](¬b)c in sequence $\mathbf s_{1}$, respectively.}
\label{fig4}
\end{figure}
\end{example}

The principle of the NegPair algorithm is shown as follows. 
Suppose patterns $\mathbf p$ and $\mathbf q$ can generate superpattern $\mathbf t$. The key-value pair arrays of $\mathbf p$, $\mathbf q$, and $\mathbf t$ in $\mathbf s_{1}$ are $A_{1}$, $A_{2}$, and $A$, respectively.  Thus, NegPair uses $A_{1}$ and $A_{2}$ to generate $A$. We first consider whether $A_{2}[k].key$ and $A_{1}[j].key$ satisfy the gap constraint or not. If they satisfy the gap constraint, i.e., $M \leq A_{2}[k].key-A_{1}[j].key-1 \leq N$, then we further consider the following case. If $e_{m-1} = null$ or $e_{m-1} \neq null$ and there is no item $e_{m-1}$ between $A_{1}[j].key$ and $A_{2}[k].key$, then there are two possible cases, where  $e_{m-1}$ is the last negative item of $\mathbf t$.

Case 1: $A[i]$ does not exist in $A$. In this case, $A[i].key = A_{2}[k].key$ and $A[i].value = A_{1}[j].value$.

Case 2: $A[i]$ already exists in $A$. In this case, $A[i].value += A_{1}[j].value$.

Example \ref{exam13} illustrates the principle of NegPair.

	\begin{example}\label{exam13}

Suppose we have a sequence $\mathbf s_{1}$ = baacaac and two patterns $\mathbf p$ = a[0,1]a and $\mathbf q$ = a[0,1]$^{\overline {b}}$c, whose key-value pair arrays $A_{1}$ and $A_{2}$ are shown in the left part of Figure \ref{fig4}. Now, we will illustrate the principle of generating $A$ for superpattern $\mathbf t$ = a[0,1]a[0,1]$^{\overline {b}}$c.

According to $A_{1}$ and $A_{2}$, we know that $A_{1}[1].key = 3$ and $A_{2}[1].key = 4$. Thus, $A[1].key = A_{2}[1].key = 4$ and $A[1].value = A_{1}[1].value = 1$, since $4-3-1 = 0$ satisfies the gap constraints [0,1], and there is no item between $d_{3}$ and $d_{4}$. We know that $A_{1}[2].key = 5$ and $A_{2}[2].key = 7$ satisfy the gap constraints and there is no ``b'' between $d_{5}$ and $d_{7}$. Therefore, $A[2].key = A_{2}[2].key = 7$ and $A[2].value = A_{1}[2].value = 1$. $A[2].value = 1+A_{1}[3].value = 2$, since $A_{1}[3].key = 6$ and $A_{2}[2].key = 7$ which satisfy the gap constraints, and there is no item between $d_{6}$ and $d_{7}$. The result of $A$ is shown in the right part of Figure \ref{fig4}.
\end{example}

NegPair is illustrated in Algorithm \ref{alg1}. 

\begin{algorithm}  
\footnotesize
\caption{NegPair} \label{alg1}
\begin{algorithmic}[1]   
 \leftline{{\bf Input:} Patterns $\mathbf p$ and $\mathbf q$, their corresponding arrays $A_{1}$ and $A_{2}$, sequence $\mathbf s$, and gap constraint $[a,b]$} 
 \leftline{{\bf Output:} Pattern $\mathbf t$ and its corresponding array $A$} 
\IF {$Suffix(\mathbf p)=Prefix(\mathbf q)$}
    \STATE $\mathbf t\leftarrow \mathbf p\oplus\mathbf q$;
    \STATE $i\leftarrow1,j\leftarrow1,k\leftarrow1$;
    \WHILE {$k\leq$  $A_{2}.$size or $j\leq$ $A_{1}.size$}
        \IF {$M\leq A_{2}[k].key-A_{1}[j].key-1\leq N$}
        \IF {$e_{m-1}=null$ or $e_{m-1}\neq null$ and there is no item $e_{m-1}$ between $A_{1}[j].key$ and $A_{2}[k].key$}
            \IF {$A[i].key=A_{2}[k].key$}
                \STATE $A[i].value\leftarrow A[i].value +A_{1}[j].value$;
            \ELSE
                \STATE $A[i].key\leftarrow A_{2}[k].key, A[i].value\leftarrow A_{1}[j].value, i\leftarrow i+1$;
	       \ENDIF
            \ENDIF
		\STATE $j\leftarrow j+1$;
	\ELSIF {$A_{2}\lbrack k\rbrack.key-A_{1}\lbrack j\rbrack.key-1<M$}
		\STATE  {$k\leftarrow k+1$;}  
	\ELSE
		\STATE $j\leftarrow j+1$;
	\ENDIF
    \ENDWHILE
\ENDIF\\
\STATE return $\mathbf t$ and $A$;
\end{algorithmic}  
\end{algorithm}

{The NegPair algorithm first checks whether \textit{Suffix}(\textbf{p}) is the same as \textit{Prefix}(\textbf{q}) (Line 1). If they are the same, then NegPair generates a superpattern pattern \textbf{t}=\textbf{p} $\oplus$ \textbf{q} (Line 2) and uses $j$ and $k$ to point to the positions of the key-value pair arrays $A_1$ and $A_2$, respectively (Line 3). When $j$ does not point to the end of $A_1$ or $k$ does not point to the end of $A_2$, the calculation begins (Line 4). If $A_2[k]$.key and $A_1[j]$.key satisfy the gap constraint (Line 5) and the negative item constraint is satisfied (Line 6), then there are two cases to update the key and the value of the key-value pair array $A$ (Lines 7 to 11). After that, $j$ is incremented by 1 (Line 13). If $A_2[k]$.key-$A_1[j]$.key-1 is smaller than $M$, then we move to the next element of \textbf{q}, i.e., $k$ is incremented by 1 (Lines 14 to 15); Otherwise, we move to the next element of \textbf{p}, i.e., j is incremented by 1 (Line 17). Finally, superpattern \textbf{t} and its key-value pair array $A$ are obtained (Line 21).}

\subsection{NSPG-Miner algorithm} \label {sect4.3}
The NSPG-Miner algorithm has four steps.

Step 1: Scan database $SDB$ to get all the frequent items and store them in $F_{1}$. 

Step 2: NSPG2 is proposed, which is used to mine NSPGs with a length of two. NSPG2 does the following work.

An enumeration tree strategy is used to generate all candidate patterns with a length of two: $\mathbf t = p_{1}[a,b]p_{2}$, where $p_{1}$ and $p_{2}$ are frequent items. Then, the support of pattern $\mathbf t$ is calculated using NegPair. If $\mathbf t$ is a frequent pattern, then $\mathbf t$ is stored in $F_{2}$ and each negative item $\overline {e} (e\in\Sigma)$ is inserted between the positive items $p_{1}$ and $p_{2}$ to generate a new candidate pattern $\mathbf t_{1} = p_{1}[a,b]^{\overline {e}} p_{2}$. If $\mathbf t_{1}$ is also a frequent pattern, then it is stored in $F_{2}$. The NSPG2 algorithm is shown in Algorithm \ref{alg2}.

\begin{algorithm}
\footnotesize
\caption{{NSPG2}}	\label{alg2}
\leftline{ {\bf Input:} Sequence database $SDB$, support threshold $\rho$, gap constraint $[a,b]$, and frequent pattern set $F_{1}$}
\leftline{ {\bf Output:} Frequent pattern set $F_{2}$}
\begin{algorithmic}[1]
\FOR {each item $p$ in $F_1$}
    \FOR { each item $q$ in $F_1$}
        \STATE $C_{2}$. add ($p[a,b]q$); //Use enumeration tree strategy to generate candidate pattern $p[a,b]q$
    \ENDFOR
\ENDFOR
\FOR {each $\mathbf t\leftarrow p_{1}[a,b]p_{2}$ in $C_{2}$}
    \STATE Calculate the support and $A$ of $\mathbf t$ using the NegPair Algorithm;
    \IF {$\mathbf t.sup \geq \rho\times L\times W$}
	\STATE Store $\mathbf t$ and its $A$ in $F_{2}$;
	\FOR {each item $e$ in $\Sigma$}
		\STATE $\mathbf t_{1}\leftarrow p_{1}[a,b]^{\overline {e}}p_{2}$;
		\STATE Calculate the support and $A_{1}$ of $\mathbf t_{1}$ using the NegPair Algorithm;
		\IF {$\mathbf t_{1}.sup \geq  \rho\times L\times W$}
			\STATE Store $\mathbf t_{1}$ and its $A_{1}$ in $F_{2}$;
		\ENDIF
	\ENDFOR
    \ENDIF
\ENDFOR
\STATE return $F_{2}$;
\end{algorithmic}
\end{algorithm}

{The NSPG2 algorithm first uses the enumeration tree strategy to generate all candidate patterns with a length of two and store them in $C_2$ (Lines 1 to 5). For each candidate pattern \textbf{t}= $p_1[a,b]p_2$ in $C_2$ (Line 6), NSPG2 uses the NegPair algorithm to calculate \textbf{t} and its key-value pair array $A$ (Line 7). If \textbf{t} is a frequent pattern (Line 8), then NegPair stores \textbf{t} and its $A$ in $F_2$ (Line 9). Moreover, each negative item $e$ is inserted between the positive items $p_1$ and $p_2$ to generate a new candidate pattern \textbf{t}$_1$ (Line 11). NSPG2 uses the NegPair algorithm to calculate the support of \textbf{t}$_1$ (Line 12). If \textbf{t}$_1$ is a frequent negative pattern (Line 13), then NegPair stores \textbf{t}$_1$ and its $A_1$ in $F_2$. Hence, all frequent positive and negative patterns with length two are stored in $F_2$.}
	
 Step 3:  NSPGm is proposed, which mines NSPGs with a length of $m+1$, where $m \geq 2$. NSPGm does the following work.
 
For any two patterns $\mathbf p$ and $\mathbf q$ in $F_{m}$, if $Suffix(\mathbf p)=Prefix(\mathbf q)$, then candidate pattern $\mathbf t$ is generated, and the support rate of $\mathbf t$ in $SDB$ is calculated using NegPair. If the support rate of $\mathbf t$ is greater than $\rho$, then $\mathbf t$ is stored in $F_{m+1}$. The NSPGm algorithm is shown in Algorithm \ref{alg3}.
 
\begin{algorithm}
\footnotesize
\caption{NSPGm}	\label{alg3}
\leftline{ {\bf Input:} 
		Sequence database $SDB$, support threshold $\rho$, gap constraint $[a,b]$, frequent pattern set $F_{m}$}
\leftline{ {\bf Output:} Frequent pattern set $F_{m+1}$}
\begin{algorithmic}[1]
\STATE $offsup\leftarrow L\times W^{m}$;
\FOR {each $\mathbf p$ in $F_{m}$}
    \FOR {each $\mathbf q$ in $F_{m}$}
	\IF{$Suffix(\mathbf p)=Prefix(\mathbf q)$}
		\STATE $sup\leftarrow0$;
    	\FOR {each $\mathbf {s}_{k}$ in $SDB$}
			\STATE $\mathbf t, A\leftarrow NegPair(\mathbf p, \mathbf q, A_{1}, A_{2}, \mathbf s_{k}, [a,b])$;
			\STATE $sup\leftarrow sup+$the sum of $A.value$;
		\ENDFOR
		\IF{$sup/offsup \geq \rho$}
			\STATE Store $\mathbf t$ and its $A$ in $F_{m+1}$;
		\ENDIF
	\ENDIF
    \ENDFOR
\ENDFOR
\STATE return $F_{m+1}$
\end{algorithmic}
\end{algorithm}

{The NSPGm algorithm first calculates  \textit {offsup}=$L\times W^m$ (Line 1). For any two patterns \textbf{p} and \textbf{q} in $F_m$ (Lines 2 and 3), if \textit{Suffix}(\textbf{p})=\textit{Prefix}(\textbf{q}) (Line 4), then NSPGm adopts NegPair to calculate the support of superpattern \textbf{t}=\textbf{p} $\oplus$ \textbf{q} (Lines 5 to 9). If pattern \textbf{t} is frequent, then \textbf{t} and its $A$ are stored in $F_{m+1}$ (Lines 10 to 12). Hence, all frequent positive and negative patterns with length $m$+1 are stored in $F_{m+1}$.}

Step 4: Iterate Step 3 until $F_{m}$ is empty.
 
{ Example \ref {exam14} illustrates the steps of NSPG-Miner.}

\begin{example}\label{exam14}
{In this example, the sequence database $SDB$ is the same as Example \ref {exam2}, i.e., $SDB=\{\mathbf s_{1}, \mathbf s_{2}\}$, $\mathbf s_{1} =$ baacaac, and $\mathbf s_{2} = $ ababccbb. Moreover, the gap constraint is [0,1] and  $\rho = 0.13$, which are the same as Example \ref {exam8}.}

{According to Step 1, we know that the frequent items are $F_{1}$=\{$a,b,c$\}.}

{According to Step 2, NSPG-Miner uses NSPG2 to mine NSPG2 with a length of two. According to Lines 1 to 5 in NSPG2, NSPG-Miner generates nine candidate patterns with length two: a[0,1]a, a[0,1]b, a[0,1]c, b[0,1]a, b[0,1]b, b[0,1]c, c[0,1]a, c[0,1]b, c[0,1]c. Taking a[0,1]a as an example, according to Example \ref {exam13}, we know that $sup(a[0,1]a, \mathbf s_{1})$=3. It is easy to know that $sup(a[0,1]a, \mathbf s_{2})$=1. Thus, $sup(a[0,1]a, SDB)$=3+1=4. We know the lengths of  $\mathbf s_{1}$ and $\mathbf s_{2}$ are 7 and 8, respectively, i.e., $L$=7+8=15. Moreover, $W$=2, since gap constraint is [0,1] and 1-0+1=2. Therefore, $\rho \times L \times W$=$0.13 \times 15\times 2$=3.9. Hence, a[0,1]a is a frequent pattern and is stored in $F_2$. According to Lines 6 to 12 in NSPG2, NSPG-Miner further calculates the supports of a[0,1]$^{\overline {a}}$a, a[0,1]$^{\overline {b}}$a, and a[0,1]$^{\overline {c}}$a, since $\Sigma$=\{a,b,c\}. Furthermore, we know that a[0,1]$^{\overline {a}}$a is a frequent negative pattern and is stored in $F_2$. Similarly, we know that a[0,1]$^{\overline {b}}$a and a[0,1]$^{\overline {c}}$a are not frequent. Moreover, we know that patterns a[0,1]c, a[0,1]$^{\overline {b}}$c, and a[0,1]$^{\overline {c}}$c are frequent. }

{According to Step 3, NSPG-Miner generates candidate patterns with a length of three using the NSPGm algorithm. Based on the five frequent patterns: a[0,1]a, a[0,1]$^{\overline {a}}$a, a[0,1]c, a[0,1]$^{\overline {b}}$c, and a[0,1]$^{\overline {c}}$c, we know that only ten candidate patterns are generated using the pattern join strategy with negative patterns, which is superior to the enumeration tree strategy and the classical negative SPM algorithms according to Example \ref {exam11}. Moreover, we use NegPair to calculate the support of each candidate pattern, and they are infrequent. Thus, $F_3$ is empty. }

{Hence, the algorithm ends according to Step 4, and the results are the same as that of Example \ref {exam11}.}

\end{example}

{The pseudo-code of NSPG-Miner is shown in Algorithm \ref{alg4}.}
 
\begin{algorithm}
\footnotesize
\caption{{NSPG-Miner}}\label{alg4}
\leftline{{\bf Input:} 	Sequence database $SDB$, support threshold $\rho$, and gap constraint $[a,b]$}
\leftline{{\bf Output:} 		Frequent pattern set $F$}
\begin{algorithmic}[1]
\FOR {each sequence $\textbf{s}_c$ in SDB}
    \FOR {each item $d_j$ in $\textbf{s}_c$}
        \STATE $A[c][d_j].key.add (j)$;
        \STATE $A[c][d_j].value.add (1)$;
    \ENDFOR
\ENDFOR
\FOR {each sequence $\textbf{s}_c$ in SDB}
    \FOR {each item $p$ in $\Sigma$}
        \STATE  $sup[p] \leftarrow  sup[p]+ A[c][d_j].size$;
    \ENDFOR
\ENDFOR
\FOR {each item $p$ in $\Sigma$}
    \IF {$sup[p]$ $\geq \rho \times$ $L$ }
        \STATE $F_1$.add ($p$);   // Item $p$ is a frequent item
    \ENDIF
\ENDFOR
\STATE $F_2\leftarrow NSPG2(SDB, \rho, [a,b], F_{1})$;
\STATE $m\leftarrow2$;
\WHILE {$F_{m}\neq null$}
    \STATE $F_{m+1}\leftarrow NSPGm(SDB, \rho, [a,b], F_{m})$;
    \STATE $m\leftarrow m+1$;
\ENDWHILE 
\STATE return $F\leftarrow F_{1}\cup F_{2}\cup\ldots\cup F_{m}$;
\end{algorithmic}
\end{algorithm}

{The NSPG-Miner algorithm first traverses SDB (Lines 1 to 6) and stores all frequent patterns in set $F_1$ (Lines 7 to 16). Then NSPG-Miner adopts NSPG2 to discover all frequent patterns with length two and stores them in $F_2$ (Line 17). NSPG-Miner sets $m$=2. While $F_m$ is not NULL, NSPG-Miner employs NSPGm to discover all frequent patterns with length $m$+1 and stores them in $F_{m+1}$ (Lines 19 to 22). Hence, all frequent patterns are stored in $F_{1}\cup F_{2}\cup\ldots\cup F_{m}$. }

 {Now, we will analyze the time and space complexities. }

\begin{theorem} \label{theorem2}
The time complexity of NSPG-Miner is $O(\frac{D\times L\times W}{|\Sigma|}+L)$, where $D$ is the total number of candidate patterns, $L$ is the length of  $SDB$, and $W=N-M+1$. 
\end {theorem}
\begin {proof}
First, we analyze the time complexity of NegPair, which is $O(\frac{L\times W}{|\Sigma|})$. The reason is as follows. NegPair needs to traverse the key-value pair arrays of $\mathbf p$ and $\mathbf q$. The sizes of the key-value pair arrays of  $\mathbf p$ and $\mathbf q$ are both $O(\frac{L}{|\Sigma|})$ on average, since the length of $SDB$ is $L$. When $k$ is updated, $j$ has to backtrack $N-M+1$ positions in Line 15 in Algorithm \ref{alg1}. Now, we analyze the complexity of NSPG-Miner. Since the length of $SDB$ is $L$, the time complexity of scanning the database $SDB$ to get all the frequent items is $O(L)$. We know that both Steps 2 and 3 use NegPair to calculate the support. There are $D$ candidate patterns, which means that NegPair runs $D$ times. Therefore, the time complexity of NSPG-Miner is $O(\frac{D\times L\times W}{|\Sigma|}+L)$.
\end {proof}

\begin{theorem}\label{theorem3}
 The space complexity of NSPG-Miner algorithm is $O(\frac{D\times L}{|\Sigma|})$.
\end {theorem}
\begin {proof}
In the NSPG-Miner algorithm, the memory usage consists of two parts: candidate patterns and all key-value pair arrays. The space complexity of candidate patterns is $O(D)$, since there are $D$ candidate patterns. The space complexity of all key-value pair arrays is $O(\frac{D\times L}{|\Sigma|})$, since the size of the key-value pair arrays is $O(\frac{L}{|\Sigma|})$. Hence, the space complexity of the NSPG-Miner algorithm is $O(D+\frac{D\times L}{|\Sigma|}) = O(\frac{D\times L}{|\Sigma|})$.
\end {proof}

\section{Experiments}\label{section5}
We will introduce the benchmark datasets and the competitive algorithms in Section \ref{sect5.1}. Section \ref{sect5.2} validates the efficiency of different strategies of the NSPG algorithm.  Section \ref{sect5.3} reports the comparisons between NSPG-Miner and the state-of-the-art PSPG mining algorithms. Section \ref{sect5.4} illustrates the comparisons between NSPG-Miner and the state-of-the-art negative sequential pattern mining algorithms. Section \ref {sectass} assesses the effect of varying the threshold and gap constraint on the number of patterns and performance. Section \ref {scalability} evaluates the scalability of  NSPG-Miner. Section \ref{sect5.5} applies NSPG-Miner to analyze biological sequences.

All experiments are conducted on a computer with an Intel(R)Core(TM)i7-8750U 2.20 GHz CPU, 16 GB of memory, and a 64-bit Windows 10 operating system. The programming environment is Visual Studio 2015.  Algorithms and datasets can be downloaded from https://github.com/wuc567/Pattern-Mining/tree/master/NSPG-Miner.

\subsection{Benchmark datasets and baseline methods}\label {sect5.1}
{To validate the performance of the NSPG-Miner algorithm, 11 datasets are selected from two categories:  DNA sequences and virus sequences, shown in Table \ref{tab2} \footnote {All datasets and algorithms can be downloaded from https://github.com/wuc567/Pattern-Mining/tree/master/NSPG-Miner. }.}

\begin{table}[ht]
\centering
\scriptsize 
\caption{{Summary of benchmark datasets}}	\label{tab2}
\begin{tabular}{cccc}	
\hline	
{Name} &{Type}  & {Source} &  {Length }\\\hline
DNA1\footnote[2] & DNA  & Homo Sapiens AL158070  & 6,000  \\
DNA2  & DNA & Homo Sapiens AL158070   & 8,000 \\
DNA3  & DNA  & Homo Sapiens AL158070   & 10,000 \\
DNA4  & DNA  & Homo Sapiens AL158070   & 12,000 \\
DNA5  & DNA  & Homo Sapiens AL158070   & 14,000 \\
DNA6  & DNA  & Homo Sapiens AL158070   & 16,000 \\
Virus1 \footnote[3]   & Virus  & Reston Ebola virus   & 18,891 \\
Virus2   & Virus  & Potato virus Y Wilga MV99   & 9,699 \\
HIV \footnote[4]  & Virus & HIV protease &  10,000 \\
SARS-1 \footnote[5]   & Virus  & Coronavirus (SARS)   & 29,751 \\
SARS-2   & Virus  & Coronavirus 2(COVID-19)   & 29,903 \\
\hline
\end{tabular}
\end{table}

\footnotetext[2]{The DNA1$-$6 datasets  can be downloaded from https://www.ncbi.nlm.nih.gov/nuccore/AL158070.11.}

\footnotetext[3] {The Virus1 and Virus2 datasets  can be downloaded from https://www.ebi.ac.uk/ena/data/view/Taxon:129003 and https://www.ebi.ac.uk/ena/data/view/Taxon:1107954, respectively.}

\footnotetext[4] {The HIV dataset can be downloaded from http://archive.ics.uci.edu/ml/machine-learning-databases/00330.}
\footnotetext[5]{The SARS-1 and SARS-2 datasets  can be downloaded from https://www.ncbi.nlm.nih.gov/nuccore/MN908947 and https://www.ncbi.nlm.nih.gov/nuccore/30271926, respectively.  }

To evaluate the performance of the NSPG-Miner algorithm, 11 algorithms in eight categories are selected as competitive algorithms. Brief descriptions of these algorithms are as follows.
  
\begin{enumerate}

\item NSPG-bf and NSPG-df: To analyze the effect of the pattern join strategy with negative patterns, we propose two variations of the NSPG-Miner algorithm called NSPG-bf and NSPG-df, which adopt queue and stack structures to generate candidate patterns using breadth-first and depth-first strategies with the enumeration tree strategy, respectively. %
Moreover, to calculate the support, NSPG-bf and NSPG-df adopt NegPair, which uses prefix and suffix patterns and their corresponding key-value arrays to calculate the support of superpatterns. In the NegPair implementation of NSPG-bf and NSPG-df, the key-value array of the prefix pattern is stored in queue and stack structures, respectively, and the key-value array of suffix pattern is the position index of each item. 

\item NSPG-intree: To analyze the effect of using NegPair for pattern matching, we propose NSPG-intree as a competitive algorithm, which uses the incomplete Nettree structure proposed in~\cite{32_Wu2014} to calculate the support. Unlike the pattern matching method of NSPG-Miner, NSPG-intree has to scan the database to calculate the supports of superpatterns.

\item NSPG-Pos: To compare the performance with positive pattern mining algorithms, we propose NSPG-Pos which only mines PSPGs. NSPG-Pos does not generate negative candidate patterns with a length of two, which means that NSPG-Pos removes the code from Lines 10 to 16 in the NSPG2 algorithm. 

\item NSPG-like: To analyze the effect of the Apriori property, we propose NSPG-like as a competitive algorithm. NSPG-like uses the method of~\cite{32_Wu2014} to calculate the number of offset sequences and uses an Apriori like property to prune candidates.

\item MAPB and MAPD~\cite{32_Wu2014}:  To report the differences between PSPG mining and NSPG mining, and to validate the efficiency of NSPG-Miner, we select MAPB and MAPD as the competitive algorithms, which are traditional algorithms that can only mine PSPGs.

\item ITM~\cite{34_Wang2019}:  To compare PSPG mining and NSPG mining, we use ITM as a competitive algorithm, which can only mine PSPGs.

\item e-NSP~\cite{41_Cao2016} and e-RNSP~\cite{43_Dong2020}: To verify the efficiency of NSPG-Miner for NSP mining, e-NSP and e-RNSP are adopted as competitive algorithms. e-NSP and e-RNSP are algorithms for mining NSPs and repetitive NSPs, respectively. 

\item ONP-Miner \cite {onpminer}: To evaluate the mining performance of NSPG-Miner, ONP-Miner is selected as a competitive algorithm, which can mine positive and negative patterns with gap constraints under the one-off condition. 

\end{enumerate}

	\subsection{Efficiency}\label {sect5.2}
{To validate the efficiency of different strategies of NSPG-Miner, we select four competitive algorithms: NSPG-bf, NSPG-df, NSPG-intree, and NSPG-like. Experiments are conducted on DNA1, DNA2, DNA3, DNA4, DNA5, DNA6, Virus1, and Virus2 datasets. All algorithms adopt the same parameters: $\rho = 0.01$ and gap constraint = [0,15]. All of the algorithms mine 137, 142, 144, 145, 151, 152, 133, and 134 patterns on DNA1, DNA2, DNA3, DNA4, DNA5, DNA6, Virus1, and Virus2, respectively. The comparisons of the running time, number of candidate patterns, and memory usage are shown in Figures \ref{fig7}, \ref{fig8}, and \ref{fig85m}, respectively.}
	\begin{figure}[h]
		\centering
		\includegraphics[width=0.45\textwidth]{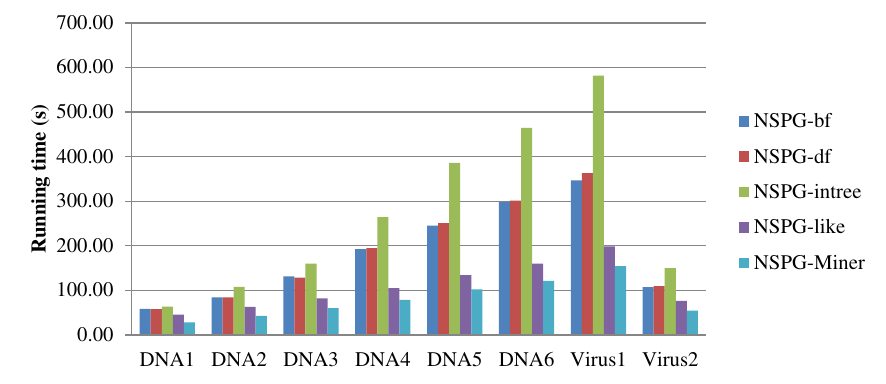}
		\caption{Comparison of running time}
		\label{fig7}
	\end{figure}
	\begin{figure}[!t]
		\centering
		\includegraphics[width=0.45\textwidth]{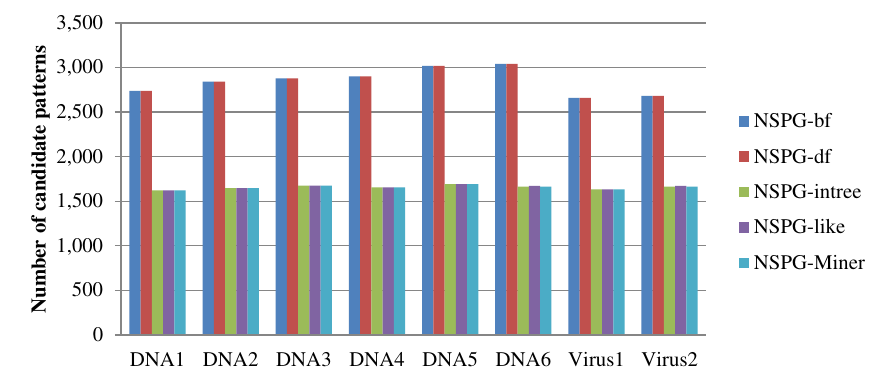}
		\caption{Comparison of number of candidate patterns}
		\label{fig8}
	\end{figure}

\begin{figure}[!t]
	\centering
	\includegraphics[width=0.45\textwidth]{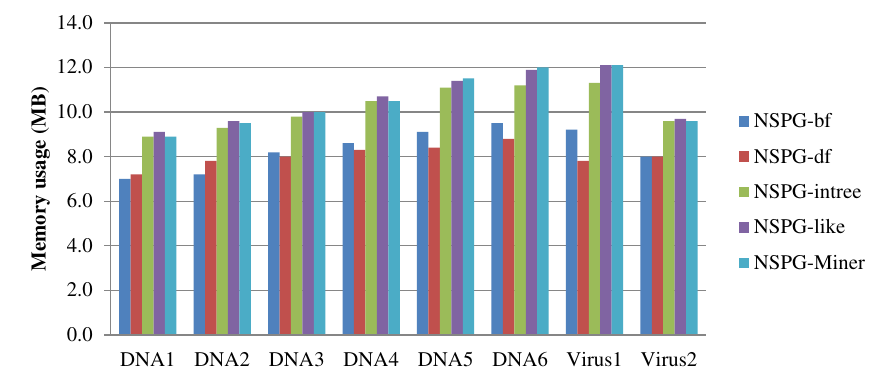}
	\caption{Comparison of memory usage}
	\label{fig85m}
\end{figure}

The results give rise to the following observations.

\begin{enumerate}
	
\item The pattern join strategy with negative patterns is more efficient than the enumeration tree strategy, since NSPG-Miner achieves a better performance than NSPG-bf and NSPG-df. From Figure \ref{fig7}, NSPG-Miner runs faster than NSPG-bf and NSPG-df. For example, on Virus1, NSPG-bf and NSPG-df take 365 s and 362 s, respectively, while NSPG-Miner only needs 217 s. The reason is as follows. To generate candidate patterns, NSPG-Miner uses the pattern join strategy with negative patterns, while NSPG-bf and NSPG-df use depth-first and breadth-first searches to realize the enumeration tree strategy, respectively. Figure \ref{fig8} shows that NSPG-Miner calculates 1,633 candidate patterns, while both NSPG-bf and NSPG-df calculate 2,660. This experimental result is consistent with that of Example \ref{exam11}. Hence, NSPG-Miner outperforms NSPG-bf and NSPG-df.

\item Adopting a key-value pair array for support calculation is more efficient than scanning the database. From Figure \ref{fig7}, NSPG-Miner runs faster than NSPG-intree. For example, on DNA5, NSPG-intree takes 385 s, while NSPG-Miner only needs 102 s. The reason is as follows. For pattern matching, NSPG-Miner adopts a key-value pair array, so it does not need to scan the database repeatedly, while NSPG-intree needs to scan the database repeatedly. Hence, NSPG-Miner outperforms NSPG-intree.

\item NSPG-Miner achieves a better performance than NSPG-like. From Figure \ref{fig7}, NSPG-Miner runs faster than NSPG-like. For example, on DNA3, NSPG-like takes about 82 s, while NSPG-Miner only needs 60 s. The reason is as follows.  NSPG-Miner only uses Lemma \ref{lemma1} to calculate the number of offset sequences, while NSPG-like takes more time to scan the database. Hence, NSPG-Miner outperforms NSPG-like.   

\item We notice that NSPG-Miner consumes more memory to mine patterns. For example, on DNA3, NSPG-Miner consumes 10 MB of memory, while NSPG-df consumes 8 MB. This phenomenon can also be observed on other datasets. The reason is as follows. To calculate the support, NSPG-Miner needs to store all key-value arrays of frequent patterns. However, NSPG-df only stores a part of the key-value arrays in the stack structure. Therefore, NSPG-Miner consumes more memory to improve the mining speed. 

\end{enumerate}

In summary, NSPG-Miner is equipped with many effective strategies.

	\subsection{Comparison with three state-of-the-art PSPG mining algorithms}\label {sect5.3}

 
{ITM \cite{34_Wang2019}, MAPB \cite{32_Wu2014}, and MAPD \cite{32_Wu2014} are three state-of-the-art algorithms that can mine the same PSPGs as NSPG-Miner. More importantly, NSPG-Miner can discover PSPGs and NSPGs at the same time. To achieve a fair comparison, NSPG-Pos, as a part of NSPG-Miner,  cannot generate negative candidate patterns. Thus, NSPG-Pos can only mine PSPGs. All algorithms are run on DNA1, DNA2, DNA3, DNA4, DNA5, DNA6, Virus1, and Virus2 datasets and adopt the same parameters: $\rho = 0.01$, and gap constraint = [0,15]. All algorithms discover 77, 79,	78,	78,	79,	79,	76, and	76 PSPGs on eight datasets, respectively. The comparisons of the running time, the number of candidate patterns, and the memory usage are shown in Figures \ref{fig9}, \ref{fig10}, and \ref{itm_mem}, respectively.} 

	\begin{figure}[h]
		\centering
		\includegraphics[width=0.45\textwidth]{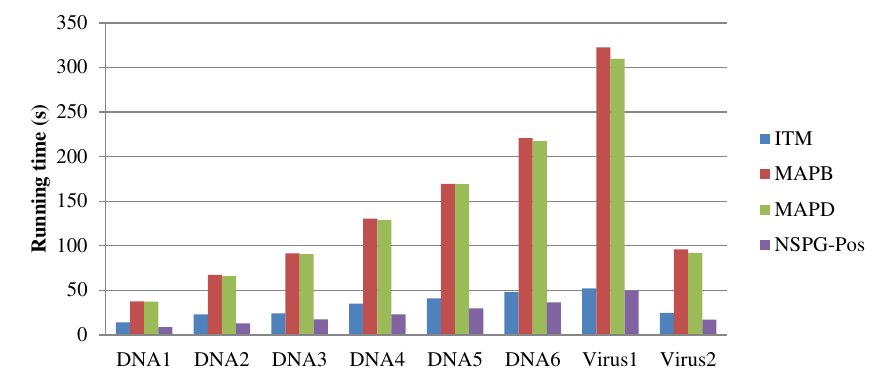}
		\caption{Comparison with three PSPG mining algorithms in terms of  running time}
		\label{fig9}
	\end{figure}
 
	\begin{figure}[h]
		\centering
		\includegraphics[width=0.45\textwidth]{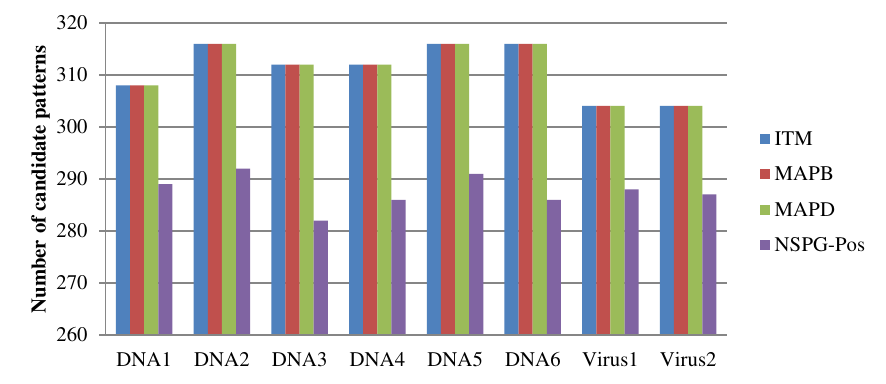}
		\caption{Comparison with  three PSPG mining algorithms in terms of the number of candidate patterns}
		\label{fig10}
	\end{figure}

\begin{figure}[h]
\centering
\includegraphics[width=0.45\textwidth]{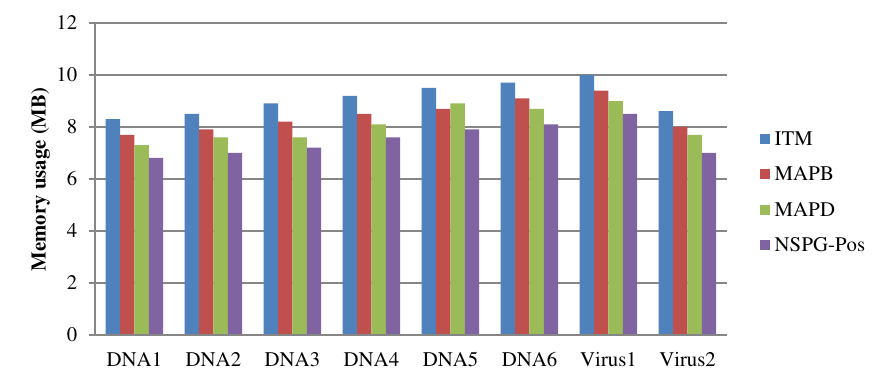}
\caption{Comparison with  three  PSPG mining algorithms in terms of  memory usage}
	\label{itm_mem}
\end{figure}
The results give rise to the following observations.

\begin{enumerate}
\item {NSPG-Pos achieves better performance than ITM, since NSPG-Pos runs faster and consumes less memory than ITM in all experiments. For example, on DNA1, Figure \ref{fig9} shows that ITM takes about 14 s, while NSPG-Pos takes about 9 s. The reason is as follows.  ITM employs a special data structure called I-Tree. The nodes in the $m$-th level represent the frequent patterns with a length of $m$, and ITM uses the $m$-th level nodes to generate and mine frequent patterns with a length of $m+1$. The I-Tree structure is more complex than our method. Therefore, ITM consumes more time and memory than NSPG-Pos. Moreover, according to Figure \ref {fig10}, ITM checks more candidate patterns than NSPG-Pos. We know that the more the candidate patterns, the longer the running time. Hence, NSPG-Pos outperforms ITM.}

\item {Compared with MAPB and MAPD, NSPG-Pos achieves better performance, since NSPG-Pos runs faster than MAPB and MAPD in all experiments. For example, on DNA3, Figure \ref{fig9} shows that MAPB and MAPD take 91 s and 90 s, respectively, while NSPG-Pos only needs about 18 s. The reason for this is as follows.   MAPB and MAPD use an incomplete Nettree structure to calculate the supports of candidate superpatterns with the same prefix pattern which needs to rescan the dataset. However, NSPG-Pos employs the key-value pair array which can effectively avoid  rescanning the dataset. Moreover, from Figure \ref {fig10}, NSPG-Pos  checks fewer candidate patterns than MAPB and MAPD. Hence, NSPG-Pos outperforms MAPB and MAPD.}
\end{enumerate}

In summary, NSPG-Pos, as a part of NSPG-Miner, has better running performance and memory consumption performance than the state-of-the-art PSPG mining algorithms.

\subsection{Comparison with three state-of-the-art NSP mining algorithms}\label {sect5.4}

{We select three state-of-the-art algorithms,  e-NSP \cite{41_Cao2016},  e-RNSP \cite{43_Dong2020}, and ONP-Miner \cite {onpminer}, as the competitive algorithms. We conduct experiments on DNA1, DNA2, DNA3, DNA4, DNA5, DNA6, Virus1, and Virus2 datasets. Since e-NSP can only mine datasets with multiple sequences, we divide these datasets into 4, 6, 8, 10, 12, 14, 16, and 18 characters per sequence, respectively, and the new datasets are called DNA1M, DNA2M, DNA3M, DNA4M, DNA5M, DNA6M, Virus1M, and Virus2M, respectively. Note that both e-NSP and e-RNSP can be considered as mining patterns with self-adaptive gap, and ONP-Miner and NSPG-Miner are mining algorithms with gap constraints. For fairness, we set the minimum gap of ONP-Miner and NSPG-Miner to 0 and the maximum gap to $l_{max}-2$, where $l_{max}$ is the length of the largest sequence in the sequence database, i.e., the gaps on these datasets are [0,2], [0,4], [0,6], [0,8], [0,10], [0,12], [0,14], and [0,16], respectively. }

{Moreover, the threshold value of e-NSP is an integer, while that of NSPG-Miner is a decimal. Therefore, it is impossible to compare the experimental results using the same threshold parameter. To be fair, we adjust the parameters of the four algorithms on different datasets to mine about 50 patterns (it is very difficult to set parameters to mine exactly 50 patterns, since the four algorithms are not top-k mining methods), which are shown in Table \ref {parameter}.}

\begin{table*}[ht]
\centering
\footnotesize
\caption{{Parameters}}	\label{parameter}

\begin{tabular}{ccccccccc}	
\hline

& 	DNA1M &	DNA2M &	DNA3M &	DNA4M &	DNA5M &	DNA6M &	Virus1M	& Virus2M \\\hline
Length of per sequence &	Len=4	&Len=6	&Len=8&	Len=10&	Len=12	&Len=14	&Len=16	&Len=18\\
Threshold of e-NSP	&280	&372	&428	&510	&650	&755	&913	&450\\
Threshold of e-RNSP	&0.02	&0.022	&0.0174	&0.0153	&0.0137	&0.0123	&0.01168	&0.01143\\
Gap &	[0,2]& [0,4]& [0,6]	&[0,8]	&[0,10]& [0,12]& [0,14]& [0,16]\\
$minsup$ of ONP-Miner&	172&	340&	540&	745&	935&	1158&	1523&	810\\
Threshold of NSPG-Miner&	$\rho$=0.0265&	$\rho$=0.022&	$\rho$=0.0188	&$\rho$=0.017	&$\rho$=0.01585&	$\rho$=0.0147 &	$\rho$=0.0131&	$\rho$=0.012\\\hline
\end{tabular}
\end{table*}

The comparisons of the number of mined patterns and negative patterns in the mined patterns and running time are shown in Figures \ref{fig12a}, \ref{fig12}, and \ref{fig12r}, respectively.

	\begin{figure}[h]
		\centering
		\includegraphics[width=0.45\textwidth]{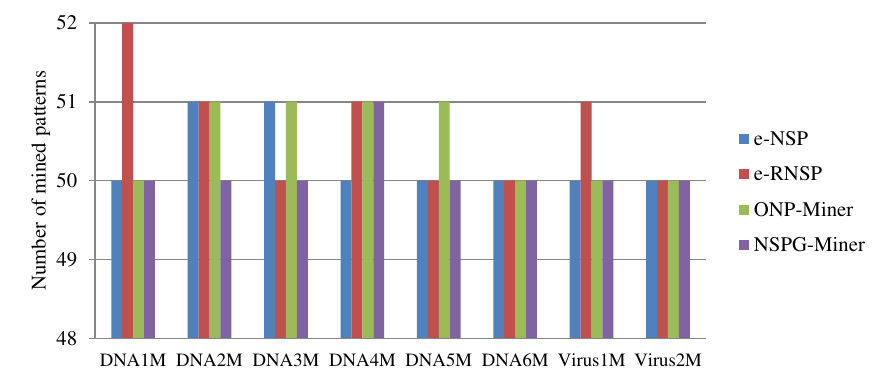}
		\caption{Comparison of  number of mined patterns}
		\label{fig12a}
	\end{figure}

	\begin{figure}[h]
		\centering
		\includegraphics[width=0.45\textwidth]{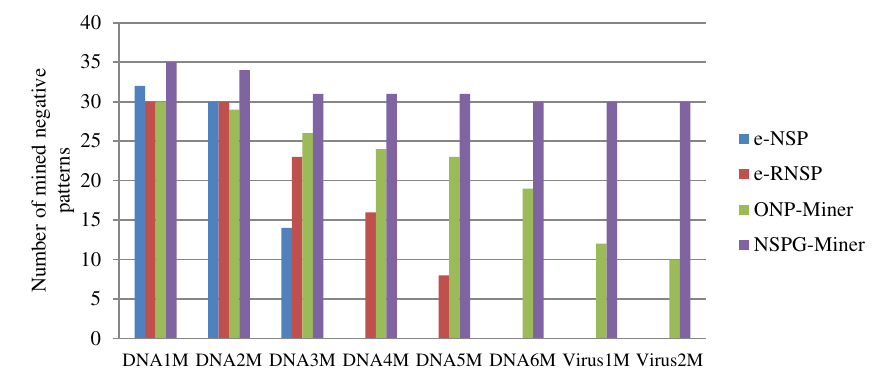}
		\caption{Comparison of  number of mined negative patterns}
		\label{fig12}
	\end{figure}
 
	\begin{figure}[h]
		\centering
		\includegraphics[width=0.45\textwidth]{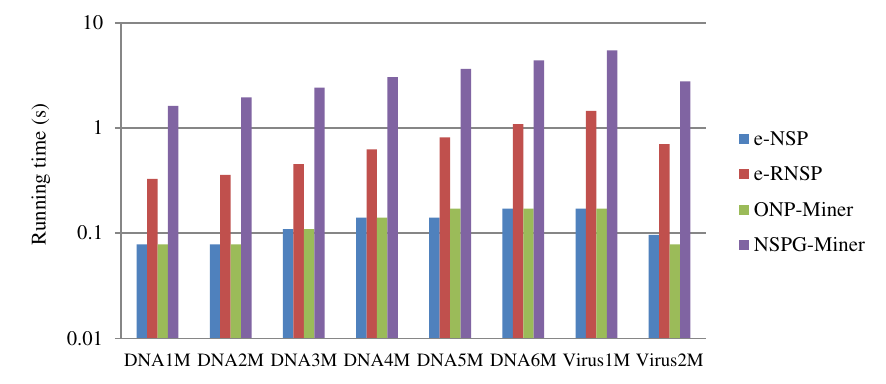}
		\caption{Comparison of running time with NSP mining algorithms}
		\label{fig12r}
	\end{figure}

The results give rise to the following observations.

\begin{enumerate}
\item {  Compared with e-NSP and e-RNSP, ONP-Miner and NSPG-Miner are more suitable for different types of datasets, especially for long sequences.  Figure \ref{fig12} shows that if the sequence length is shorter than 12, e-NSP and e-RNSP can discover negative patterns. Otherwise, e-NSP and e-RNSP cannot discover negative patterns. For example, on DNA6M, the sequence length is 14, and e-NSP and e-RNSP fail to mine NSPs, while NSPG-Miner mines 30 NSPGs. The reason is as follows. Since the negative containment given by e-NSP and e-RNSP is too strict, when the average length of the dataset is long, they cannot mine NSPs. However, ONP-Miner and NSPG-Miner are more flexible in pattern matching. Accordingly, it can play a stable role for datasets with different average lengths.}

\item  {NSPG-Miner can find more negative patterns than e-NSP, e-RNSP, and ONP-Miner. For example, on DNA3M, from Figure \ref{fig12}, e-NSP, e-RNSP, ONP-Miner, and NSPG-Miner discover 14, 23, 26, and 32 NSPs, respectively. e-NSP and e-RNSP put too many constraints on NSPs, which may cause the missing of NSPs. Meanwhile, ONP-Miner and NSPG-Miner alleviates the constraints on negative sequences; therefore, more comprehensive negative sequences can be found. More importantly, the results indicate that NSPG-Miner has a better ability to mine negative patterns than ONP-Miner.}

\item {According to Figure \ref {fig12r}, the weakness of the NSPG-Miner algorithm cannot be neglected, that is, the NSPG-Miner algorithm is slower than e-NSP, e-RNSP, and ONP-Miner. For example, on DNA4, NSPG-Miner algorithm is about 22 times slower than e-NSP, about 5 times slower than e-RNSP,  and about 22 times slower than ONP-Miner, since the running time of e-NSP, e-RNSP, ONP-Miner, and NSPG-Miner are 0.17s, 0.63s, 0.17s and 18.8s, respectively. The reason is as follows. e-NSP is the fastest algorithm, since this mining method does not consider the pattern repetition in a sequence. Thus, it is easy to calculate the support and discover NSPs. Based on e-NSP, e-RNSP was proposed to discover repetitive patterns which means that this mining method calculates the support of a pattern in a sequence. 
Our mining method is more complex than e-RNSP, since our mining method considers gap constraints, while e-RNSP does not.  Although ONP-Miner considers the pattern repetition and gap constraints at the same time, ONP-Miner employs a heuristic strategy to approximately mine the positive and negative patterns. Moreover, ONP-Miner adopts the one-off condition which means that each character of a sequence can be used at most once. However, our mining method is a completeness algorithm and further needs to determine that there are no specific negative items in the gap constraint. More importantly, in our mining method, each character can be reused, which is very time-consuming. Hence, NSPG-Miner is slower than e-NSP, e-RNSP, and ONP-Miner.}
\end{enumerate}


{In summary, compared with the state-of-the-art algorithms, NSPG-Miner can discover negative patterns on long sequences and can mine more negative patterns, but is slower than e-NSP, e-RNSP,  and ONP-Miner. }

\subsection {Influence of parameters} \label {sectass}

To provide a comprehensive evaluation, we also assess the effect of varying the threshold and gap constraint on performance.

\subsubsection {Influence of different thresholds}

To analyze the influence of different thresholds on the performance of the NSPG-Miner algorithm, we select DNA1 as the experimental dataset and select the NSPG-bf, NSPG-df, NSPG-intree, and NSPG-like algorithms as competitive algorithms. The gap constraint of the experiment is set to [0, 15]. All these algorithms discover 137, 121, 90, 87, 60, and 46 frequent patterns with $ \rho $=0.01, $ \rho $=0.012, $ \rho $=0.014, $ \rho $=0.016, $ \rho $=0.018 and $\rho$=0.02, respectively. The comparisons of the running time, number of candidate patterns, and memory usage are shown in Figures \ref{thr5r}, \ref{thr5c}, and \ref{thr5m}, respectively.

	\begin{figure}[h]
		\centering
		\includegraphics[width=0.45\textwidth]{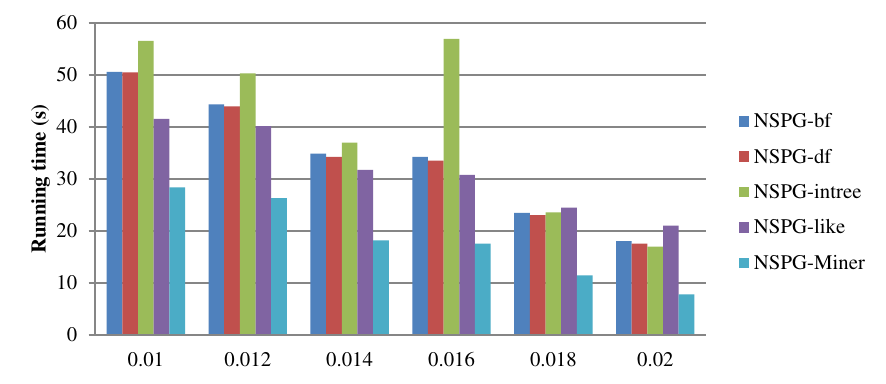}
		\caption{Comparison of running time under different thresholds}
		\label{thr5r}
\vspace{-0.5cm}
\end{figure}
 
	\begin{figure}[h]
		\centering
		\includegraphics[width=0.45\textwidth]{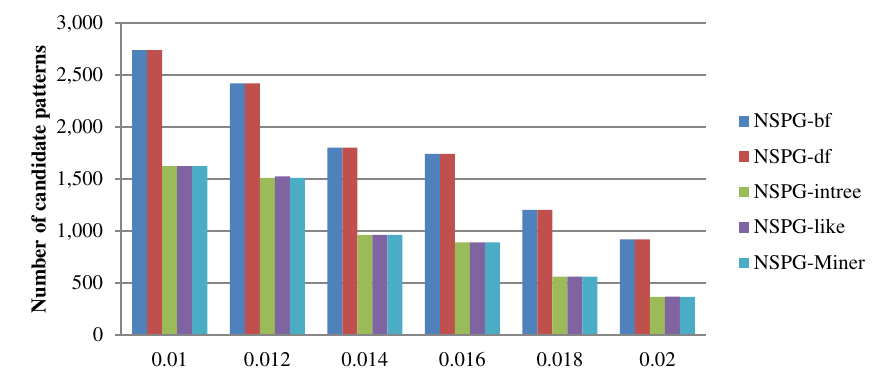}
		\caption{Comparison of number of candidate patterns under different thresholds}
		\label{thr5c}
  \vspace{-0.5cm}
	\end{figure}

	\begin{figure}[h]
		\centering
		\includegraphics[width=0.45\textwidth]{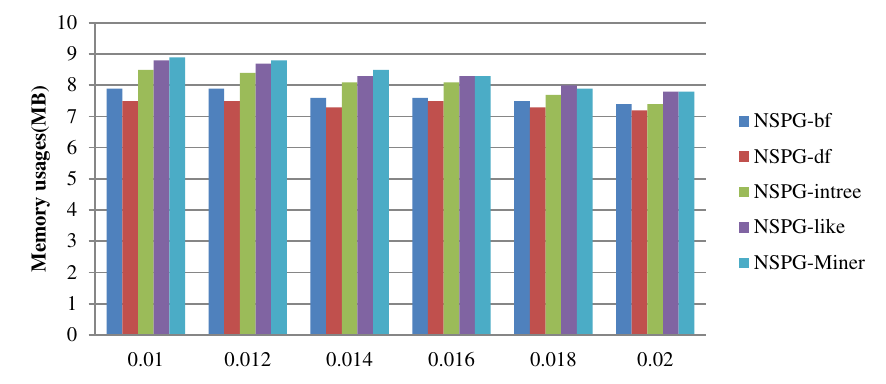}
		\caption{Comparison of memory usage under different thresholds}
		\label{thr5m}
  \vspace{-0.5cm}
	\end{figure}

The results give rise to the following observations.

\begin{enumerate}
\item As $\rho$  increases, the number of frequent patterns mined by all algorithms decreases, and the running time and the number of candidate patterns also decrease. Moreover, the memory usage also decreases. For example, when threshold $\rho$=0.01, all algorithms mine 137 frequent patterns, while when $\rho$=0.02, all algorithms discover 46 frequent patterns. Moreover, when $\rho$=0.01, NSPG-Miner runs for 28.33 s, and when $\rho$=0.02, NSPG-Miner runs for 7.81 s. For a value of $\rho$=0.01, NSPG-Miner checks 1623 candidate patterns, while for $\rho$=0.02, NSPG-Miner checks 367 patterns. Meanwhile, the memory usage decreases from 8.9 MB to 7.8 MB. The reason is as follows. For a given pattern and sequence database, the number of all offset sequences is constant. As the threshold increases, the support increases as well. Therefore, fewer patterns become candidate patterns and the number of mined patterns decreases. Hence, the running time and memory usage also decrease.

\item NSPG-Miner has better running performance than other competitive algorithms no matter what the threshold is. According to Figure \ref{thr5r}, when threshold $\rho$=0.014, compared with the running time of the other four algorithms, the running time of the NSPG-Miner algorithm is 18.14s, which is the shortest. This phenomenon can be found under different thresholds. The experimental results validate that different thresholds do not affect the superiority of our algorithm.
\end{enumerate}

\subsubsection {Influence of different gap constraints}
To analyze the influence of different gap constraints on the performance of the NSPG-Miner algorithm, we select DNA1 as an experimental dataset and select the NSPG-bf, NSPG-df, NSPG-scan, and NSPG-like algorithms for comparison. The parameter of the experiment is $ \rho $=0.01. All these algorithms mine 262, 192, 159, 145, 140, and 136 patterns for gap constraints [0,6], [0,8], [0,10], [0,12], [0,14], and [0,16], respectively. 
Although the gaps are different, all these experiments discover 77 PSPGs, which means that these experiments discover 185, 115, 82, 68, 63, and 59 NSPGs, respectively. The comparisons of running time, number of candidate patterns, and memory usage are shown in Figures \ref{gap5r}, \ref{gap5c}, and \ref{gap5m}, respectively.
	\begin{figure}[h]
		\centering
		\includegraphics[width=0.45\textwidth]{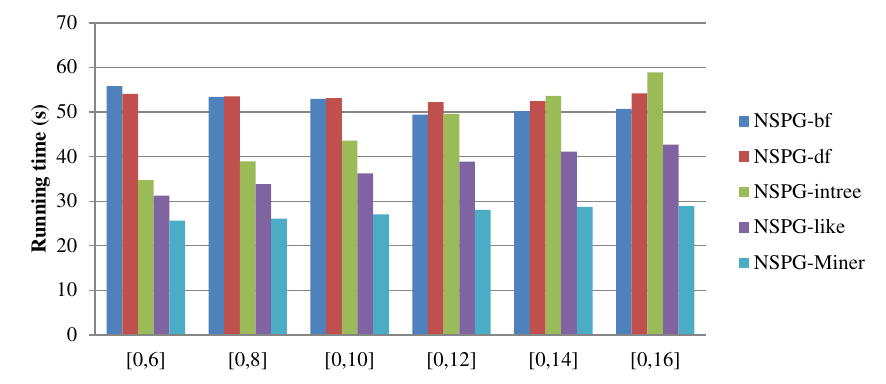}
		\caption{Comparison of running time under different gap constraints}
		\label{gap5r}
	\end{figure}

	\begin{figure}[h]
		\centering
		\includegraphics[width=0.45\textwidth]{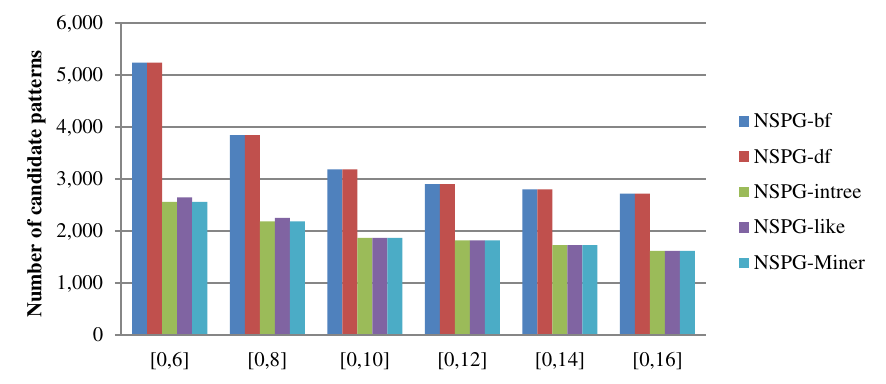}
		\caption{Comparison of number of candidate patterns under different gap constraints}
		\label{gap5c}
	\end{figure}

	\begin{figure}[h]
		\centering
		\includegraphics[width=0.45\textwidth]{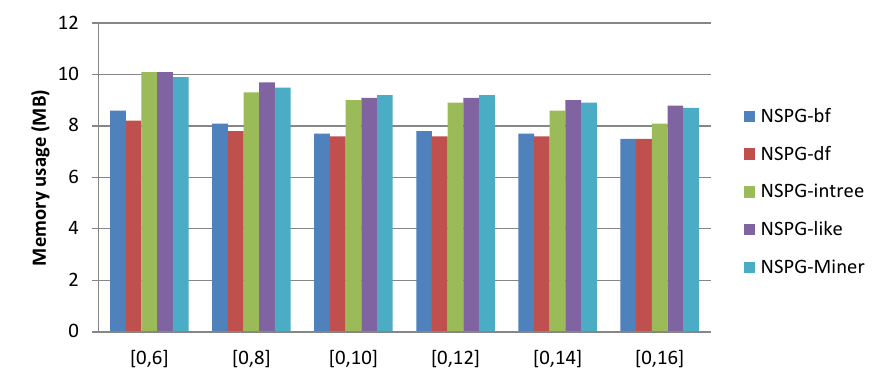}
		\caption{Comparison of memory usage under different gap constraints}
		\label{gap5m}
	\end{figure}

The results give rise to the following observations.

\begin{enumerate}
\item With the increase of the gap constraint, the number of PSPGs mined remains unchanged, the number of mined patterns, candidate patterns, and memory usage decrease, but the running time increases. For example, when the gap constraint increases from [0, 6] to [0,16], the number of mined patterns decreases from 262 to 136. Among them, 77 patterns are PSPGs, while the number of NSPGs decreases from 185 to 59. Moreover, according to Figures \ref{gap5r} to \ref{gap5m}, we know that the number of candidate patterns decreases from 2559 to 1616, and that the memory usage decreases from 9.9 to 8.7 MB, while the running time increases from 25.66 to 28.91s. The reasons are as follows. We know that NSPG-Miner discovers the negative patterns. For a negative pattern, suppose that it has a negative item, which means that the negative item does not occur in the gap. Obviously, the smaller the gap constraint, the fewer items in the gap, and the corresponding probability of missing items increases, thus increasing the number of negative patterns. Therefore, with the increase of the gap constraints, the number of mined NSPGs, the number of candidate patterns, and the memory usage decrease. However, according to Theorem \ref{theorem1}, the time complexity of NSPG-Miner is positively correlated with the size of the gap constraints and the number of candidate patterns. We know that the gap constraint is expanded by 2.67 times, while the number of candidate patterns is reduced by 1.58 times. Thus, the increase of the gap constraints is more significant than the decrease of the number of candidate patterns. Hence, the running time increases.

\item NSPG-Miner is faster than the compared algorithms no matter what the gap constraints is. According to Figure \ref{gap5r}, when the gap constraint is [0, 10], NSPG-Miner is faster than the other four algorithms, with a running time of 27.02 s, which is the shortest. This phenomenon is observed for different gap constraints. The experimental results validate that different gap constraints do not affect the superiority of NSPG-Miner over the compared algorithms.
\end{enumerate}

In summary, different thresholds and gap constraints affect the performance of NSPG-Miner, but does not affect its superiority.

\subsection{Scalability}\label {scalability}

{To evaluate the scalability of NSPG-Miner, we evaluate NSPG-Miner on six datasets with sizes of 100 KB, 150 KB, 200 KB, 250 KB, 300 KB, and 350 KB, which are 10, 15, 20, 25, 30, and 35 times of the HIV dataset, respectively. We select four competitive algorithms: NSPG-bf, NSPG-df, NSPG-intree, and NSPG-like. All algorithms adopt the same parameters: $\rho = 0.008$ and gap constraint = [0,8]. The comparisons of the running time and memory usage are shown in Figures \ref{scr} and \ref{scm}, respectively.}
 
	\begin{figure}[h]
		\centering
		\includegraphics[width=0.45\textwidth]{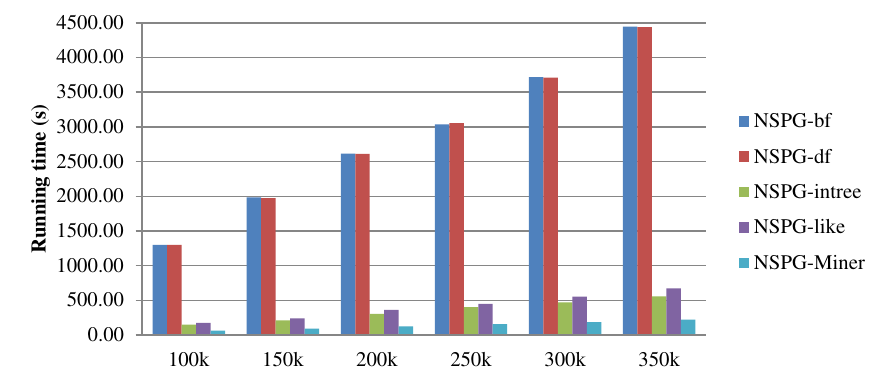}
		\caption{Scalability of NSPG-Miner on running time}
		\label{scr}
	\end{figure}

	\begin{figure}[h]
		\centering
		\includegraphics[width=0.45\textwidth]{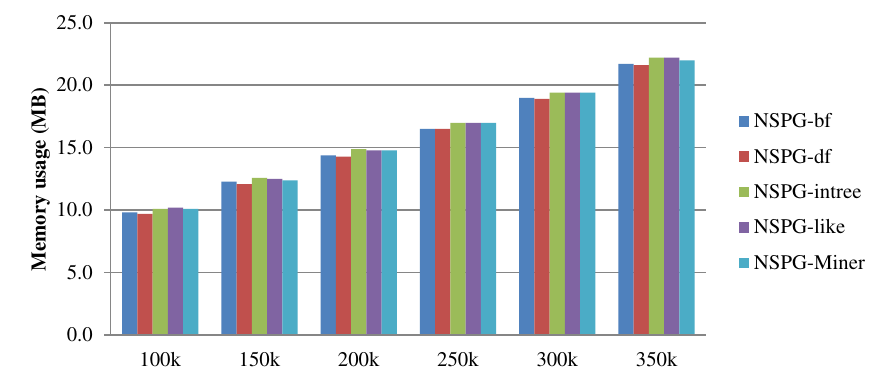}
		\caption{Scalability of NSPG-Miner on memory usage}
		\label{scm}
	\end{figure}

The results give rise to the following observations.

\begin{enumerate}

\item {With the increase of the dataset size, the running time also increases. For example, when the dataset size is 100 kb, NSPG-Miner runs about 62 s; and when the dataset length is 200 kb, NSPG-Miner runs for 125 s. This phenomenon can also be seen for the compared algorithms. The reason is shown as follows. According to Theorem \ref{theorem1}, the time complexity of NSPG-Miner is positively correlated with the dataset length. Therefore, the running time will also increase with the increase of the dataset size. Moreover, we notice that NSPG-Miner is faster than the compared algorithms, especially faster than NSPG-bf. For example, NSPG-Miner runs about 93 s on the 150 kb dataset, while NSPG-bf runs for 1980 s. The reason is that the HIV dataset has 20 different characters and NSPG-bf has to enumerate each character. As mentioned in Example \ref{exam11}, compared with the enumeration tree strategy, the pattern join strategy with negative patterns can effectively reduce the number of candidate patterns. Hence, NSPG-Miner runs faster than NSPG-bf.}

\item {As dataset size increases, the memory usage also increases. For example, NSPG-Miner consumes 10 MB on 100 KB dataset and consumes 19 MB on 300 KB dataset. Moreover, NSPG-Miner consumes more memory than other competitive algorithms. The reason is that NSPG-Miner has to store all frequent patterns and their key-value pair arrays which indicates that NSPG-Miner consumes more memory to improve the mining speed.}
\end{enumerate}

By observing  Figures \ref{scr} and \ref{scm}, we can draw the conclusion that NSPG-Miner has excellent scalability, since the mining performance does not degrade much as the dataset size increases.

\subsection{Case study}\label {sect5.5}
{Ref. \cite{onpminer} has demonstrated that negative SPM with gap constraints under the one-off condition can be used for traffic flow analysis. Our method is also a negative SPM with gap constraints, which can also conduct similar tasks. To better demonstrate the mining ability, we select a bioinformatics analysis application. }

{In 2019, COVID-19 (SARS-2) spread around the world, and in 2003, SARS (SARS-1), which is also a coronavirus, dealt a heavy blow to human beings. COVID-19 is not only the same species as the SARS virus we know, but also officially recognized as a close relative. Many researchers have studied SARS-2 from different aspects. For example, Nawaz et al. \cite{covid19} adopted sequential pattern mining to find hidden patterns that can be used to examine the evolution and variations in COVID-19 strains. Some studies have analyzed the similarity of the two from the perspective of mining maximal patterns \cite{30_Li2021} and closed patterns \cite{29_Wu2020}. However, these studies cannot analyze these two viruses from the perspective of the missing items. To overcome this shortage, this paper adopts NSPG-Miner to analyze SARS-1 and SARS-2 from the perspective of both PSPGs and NSPGs. Moreover, we select ONP-Miner \cite{onpminer} as a competitive algorithm. For fairness, the two algorithms adopt the same gap constraint with [0,15] and for NSPG-Miner, we set $\rho$= 0.006, and for ONP-Miner, we set \textit {minsup}=3000. Therefore, the two algorithms can discover similar number of patterns.  The comparison of mining results is shown in Table \ref {sars}.}

	\begin{table}[h!t]
		\centering
\scriptsize
		\caption{Comparison of mining results}
		\label{sars}
			\begin{tabular}{ccccccc}		\hline
    & \multicolumn{3}{c}{ONP-Miner}   &	\multicolumn{3}{c}{NSPG-Miner}
				\\\hline
   & All patterns	&Positive patterns	&Negative patterns	&All patterns	&Positive patterns &	Negative patterns\\\hline
SARS-1	&159	&144	&11	&158	&83&	75\\
SARS-2	&158	&137	&17	&183	&82&	101\\
Intersection	&148	&137	&11	&149&	76&	73\\
Union	&169	&152	&17	&192	&89&	103\\
Difference	&21	&15	&6	&\textbf{43}	&13	&\textbf{30}\\
				\hline
			\end{tabular}
	\end{table}

The results give rise to the following observations.
\begin{enumerate}
\item According to Table \ref {sars}, NSPG-Miner has better negative pattern mining ability than ONP-Miner, since the two algorithms discover almost the same number of patterns, while NSPG-Miner mines significantly more negative patterns than ONP-Miner. For example, ONP-Miner discovers 159 patterns on SARS-1, and NSPG-Miner mines 158 patterns. Thus, the two algorithms discover almost the same number of patterns. However, ONP-Miner only discovers 11 NSPGs on SARS-1, while NSPG-Miner mines 75. Hence, NSPG-Miner can effectively mine more negative patterns than ONP-Miner.

\item {It is easier to find the differences between SARS-1 and SARS-2 sequences from the negative pattern mining results of NSPG-Miner. Comparing the frequent patterns of SARS-1 and SARS-2, ONP-Miner only discovers 21 different frequent patterns, of which 15 are positive patterns and 6 are negative patterns, while NSPG-Miner mines 43 different frequent patterns, of which 13 are positive patterns and 30 are negative patterns. We can see that the number of positive patterns is almost the same, which is almost consistent with the research results of mining maximal patterns \cite{30_Li2021} and closed patterns \cite{29_Wu2020}. More importantly, the number of different negative patterns discovered by NSPG-Miner is 30/13=2.3 times that of different positive patterns, and five times that of negative patterns of ONP-Miner.  The above results indicate that NSPG-Miner outperforms ONP-Miner, since the patterns mined by NSPG-Miner show more differences within similar viruses, which is beneficial for biologists when analyzing viruses.}
\end{enumerate}
In summary, negative SPM can provide more information than positive SPM, and NSPG-Miner can provide more valuable information than ONP-Miner, which is more beneficial for the research on COVID-19.

\section{Conclusion and future work}\label{section6}
To find patterns with missing items, we study a novel problem of mining repetitive negative sequential patterns with gap constraints and propose an effective algorithm, called NSPG-Miner, which discovers both PSPGs and NSPGs.  NSPG-Miner performs two major steps: candidate pattern generation and support calculation. In the candidate pattern generation stage, we propose a pattern join strategy with negative patterns to generate negative candidate patterns which can effectively prune the redundant patterns and avoid invalid calculations. In the process of support calculation, we propose a NegPair algorithm that employs a key-value pair array to calculate the support of each candidate pattern which can effectively use the results of the subpatterns to calculate the supports of the superpatterns and deal with the gap constraints and the negative items at the same time. To validate the performance of NSPG-Miner, 11 competitive algorithms and 11 datasets are selected. The experimental results show that NSPG-Miner performs better than other competitive algorithms. More importantly, according to the negative pattern mined by NSPG Miner, it is easier to detect differences between virus sequences, since NSPG-Miner can discover more valuable information than the state-of-the-art negative SPM methods.

 {In this paper, we investigate the task of negative SPM with gap constraints to find frequent positive and negative patterns with gap constraints. However, this research has some limitations.}

\begin {enumerate}
\item {Our mining method is stricter than other negative sequential pattern mining methods, such as e-NSP \cite {41_Cao2016} and e-RNSP \cite {43_Dong2020}. Therefore, our mining method is far slower than e-NSP and e-RNSP. How to improve the algorithm running performance is worth exploring.}

\item {In this paper, we investigate a novel problem of mining negative sequential patterns with gap constraints and have employed this mining method to analyze the differences between two virus sequences. Moreover, negative sequential pattern mining with gap constraints can also be applied in many fields, such as customer behavior analysis for detecting abnormal patterns, traffic flow analysis for predicting trends, and the relationship analysis between students' behaviors and their grades. These applications are worth investigating in the future.}

\end {enumerate}



\section*{Acknowledgement}
This work was supported by the National Natural Science Foundation of China (62372154, 52477232, 62120106008). 

\end{document}